\documentclass[article]{IEEEtran}
\usepackage{epsfig,latexsym,amsmath,amssymb,setspace,cite,color,graphicx,algorithm,cases}
\usepackage[caption=false,font=footnotesize]{subfig}
\usepackage[percent]{overpic}

\usepackage{physics}
\usepackage{algcompatible}

\usepackage{amsmath}
\usepackage{amsfonts}
\usepackage{amssymb}
\usepackage{dsfont}
\usepackage{bm}
\usepackage{amsthm}
\usepackage{newlfont}
\usepackage{float}
\usepackage{hyperref}
\usepackage{algorithm}
\usepackage{enumerate}
\usepackage{chngcntr}
\usepackage{mathtools}
\usepackage{breqn}
\usepackage{stmaryrd}
\usepackage{cite}
\usepackage[yyyymmdd]{datetime}
\hypersetup{
    colorlinks=true, 
    linkcolor= blue, 
    citecolor=blue 
}

\usepackage{algpseudocode}

\newcounter{algsubstate}

\newenvironment{algsubstates}
  {\setcounter{algsubstate}{0}
   \renewcommand{\State}{
     \refstepcounter{algsubstate}
     \Statex {\footnotesize\alph{algsubstate}:}\space}}
  {}
\newcommand\numberthis{\addtocounter{equation}{1}\tag{\theequation}}

\begin{document}
\sloppy
\allowdisplaybreaks[1]

\newtheorem{thm}{Theorem} 
\newtheorem{lem}{Lemma}
\newtheorem{prop}{Proposition}
\newtheorem{cor}{Corollary}
\newtheorem{defn}{Definition}
\newcommand{\remarkend}{\IEEEQEDopen}
\newtheorem*{remark}{Remark}
\newtheorem{rem}{Remark}
\newtheorem{ex}{Example}
\newtheorem{pro}{Property}

\newenvironment{example}[1][Example]{\begin{trivlist}
\item[\hskip \labelsep {\bfseries #1}]}{\end{trivlist}}

\renewcommand{\qedsymbol}{ \begin{tiny}$\blacksquare$ \end{tiny} }

\renewcommand{\leq}{\leqslant}
\renewcommand{\geq}{\geqslant}

\title {\huge{Dual-Source SPIR over a noiseless MAC without Data Replication or Shared Randomness}}

\author{\IEEEauthorblockN{R\'emi A. Chou} \thanks{R. Chou is with the Department of Computer Science and Engineering, The University of Texas at Arlington, Arlington, TX. This work was supported in part by NSF grants CCF-2047913 and CCF-2401373. E-mail: remi.chou@uta.edu. A preliminary version of this work has been presented at the 2021 IEEE Information Theory Workshop \cite{chou2020} and the 2024 IEEE International Symposium on
Information Theory (ISIT)~\cite{chou2024dual}.}}

\maketitle
\begin{abstract}
Information-theoretically secure Symmetric Private Information Retrieval (SPIR) from a single server is known to be infeasible over noiseless channels. Known solutions to overcome this infeasibility involve additional resources such as database replication, shared randomness, or noisy channels. 
In this paper, we propose an alternative approach for achieving SPIR with information-theoretic security guarantees, without relying on shared randomness, noisy channels, or data replication. Specifically, we demonstrate that it is sufficient to use a noiseless binary adder multiple-access channel, where inputs are controlled by two non-colluding servers and the output is observed by the client, alongside a public noiseless communication channel between the client and the servers. Furthermore, in this setting, we characterize the optimal file rates, i.e., the file lengths normalized by the number of channel uses, that can be~transferred.
\end{abstract} 

\section{Introduction}

Consider a client who wishes to download a file from a server such that (i) their file selection is kept private from the server and (ii) the client does not learn any other information beyond the selected file. This setting is referred to as Symmetric Private Information Retrieval (SPIR). SPIR is also known as oblivious transfer \cite{rabin1981exchange,goldreich1987play}, which is anterior to SPIR and a fundamental cryptographic building block~{\cite{kilian1988founding}}.

Under information-theoretic security guarantees, which is the focus of this paper, it is well known, e.g.,~\cite{crepeau1997efficient}, that SPIR between a client and a single server is infeasible over a noiseless communication channel. To overcome this impossibility result, two approaches have  previously been considered. The first approach consists in considering additional resources at the client and server in the form of correlated randomness, which could, for instance, be obtained through a noisy channel between the client and the server. Specifically, for some classes of noisy channels, SPIR is  known to be feasible under information-theoretic security, e.g.,  \cite{crepeau1988achieving,crepeau1997efficient,nascimento2008oblivious,ahlswede2013oblivious,Shekofteh2025B}. The second approach consists in replicating data in multiple servers and assuming that the servers share randomness, e.g., \cite{gertner2000protecting,sun2018capacity,chen2020capacity,wang2019symmetric}. Note that with this second approach, a necessary assumption is that only a strict subset of servers can collude against the client, otherwise SPIR is infeasible as the setting reduces to the case of SPIR between a client and a single server over a noiseless channel. This paper aims to demonstrate the feasibility of a third approach.

In this paper, we propose to perform SPIR under information-theoretic security guarantees without shared randomness, a noisy channel, or data replication. Specifically, we consider SPIR between one client and two non-colluding servers, using a noiseless multiple-access channel, whose inputs are controlled by the servers and output is observed by the client, and a public communication channel between the client and the servers. The two servers have independent content so that no data is replicated. In this setting, the client aims to obtain one file from each server such that (i) the file selection remains private from the servers, (ii)~the unselected files remain unknown to the receiver, and (iii) one server does not learn anything about the content of the other server. As formally described in Section~\ref{secps}, in our setting,  the servers and the client can communicate over a noiseless public channel and have access to a noiseless binary adder multiple-access channel, but neither noisy channels nor pre-shared correlated randomness are available at the parties. While information-theoretically secure SPIR is impossible if the client engages in independent protocols with each of the servers, i.e., with one server at a time and without any intervention of the other server,  we show that a multiuser protocol between the client and the two servers can enable information-theoretically secure SPIR with positive rates for both servers simultaneously by leveraging the binary adder multiple-access channel, instead of shared randomness, a noisy channel, or data replication. The main idea of the achievability scheme is to use the communication of one sender as noisy resource for the other sender and the receiver. Perhaps surprisingly, this can be accomplished without time sharing between the two servers, i.e., the client can retrieve a file from each server simultaneously. More specifically, we begin by designing an achievability scheme for the case where each server holds only two files, and then build on this scheme to handle the case where each server has multiple files. A key challenge lies in the proof of the converse to establish the optimality of our approach. As a problem of independent interest, in the converse proof, we solve an optimization problem involving the distribution that maximizes the conditional entropy of two Bernoulli random variables given their sum. Our main results fully characterize the capacity region of this setting, defined as the set of all achievable rate pairs $(R_1, R_2)$, where $R_i$ denotes the retrieval rate from Server $i$, given by the ratio of the requested file size to the total number of channel uses. Unlike prior works on SPIR with replicated or coded data, which use the total download cost from the servers as the performance metric, our setting considers the number of channel uses of the binary adder multiple-access channel as the performance metric.

Note that a noiseless binary adder multiple access channel can be modeled by physical-layer schemes where simultaneous transmissions on a shared medium yield an output equal to the number of active users. Each user activity is represented as a binary input (0 if inactive, 1 if active) such that the receiver observes 0 when no one transmits, 2 when both transmit, and 1 when only one is active. In the last case, the receiver knows a single user sent a signal but cannot tell which one. This effect could be realized in several ways. For example, in simultaneous wireless transmissions with energy detection, each transmitter either emits a short signal or remains silent on the same frequency, allowing an energy detector at the receiver to determine how many transmitters were active. In power line communication  backscatter, devices modulate the reflection of an existing signal on the shared wiring so that the combined reflections reveal the number of active devices. In optical channels, each transmitter could emit light if active, and the signals can be combined at a photodetector so that the measured intensity indicates how many transmitters were active. Realizing this effect in practice requires at least transmitter synchronization, controlled signal amplitudes, and limited channel variability to preserve linear superposition.

\emph{Related works}: The crux of our approach is the use of a noiseless multiple-access channel.  By contrast, other recent works have considered SPIR over \emph{noisy} multiple-access channels, including \cite{shekofteh2024spir, aghaee2025oblivious,elimelech2025spir}. Compared with \cite{aghaee2025oblivious}, where the privacy constraint concerns the client selection with respect to individual servers, our work considers privacy of the client selection with respect to both servers jointly. Moreover, while~\cite{aghaee2025oblivious} focuses on two files, we consider an arbitrary number of files.
In \cite{elimelech2025spir}, the setting involves an arbitrary number of non-communicating and replicated servers over a Gaussian multiple-access channel, and    joint channel-SPIR schemes that do not require common randomness among the servers are proposed. Our setting also does not rely on common randomness, but unlike \cite{elimelech2025spir}, it avoids data replication and instead considers two servers with independent content.   In \cite{shekofteh2024spir}, for a specific class of noisy multiple-access channels, positive rates are achievable even when all servers collude, in the scenario where the file selection index is identical across servers, as may happen when files are distributed across multiple servers.

Similar to our setting, a noiseless binary adder multiple-access channel was also studied in \cite{pei2023capacity}. The main difference  is that \cite{pei2023capacity} considers replicated servers, whereas we do not allow data replication. Another difference is that  \cite{pei2023capacity}  focuses on the case of two files, whereas we consider a setting with an arbitrary number of files. We also note that prior related work includes  the study of PIR with non-communicating, replicated servers  over  a multiple-access channels formed by the product of  
 noisy point-to-point channels \cite{banawan2019noisy}, as well as over a Gaussian multiple-access channel in     \cite{shmuel2021private} and in \cite{elimelech2025pir}, in the case where  server responses are public.

The remainder of the paper is organized as follows. We formally introduce the setting in Section \ref{secps} and state our main results in Section \ref{secres}. We prove our converse and achievability results in Section \ref{secproof2} and~\ref{secproof3}, respectively. Finally, we provide concluding remarks in Section \ref{secconc}.

\section{Problem Statement} \label{secps}
Notation:
For $a,b\in \mathbb{R}$, define $\llbracket a,b\rrbracket \triangleq [ \lfloor a \rfloor, \lceil b \rceil]\cap \mathbb{N}$. For any $z \in \{ 1,2 \}$, define $\bar{z} \triangleq 3-z$. For a vector $X^n$ of length $n$, for any set $\mathcal{S} \subseteq \llbracket 1,n \rrbracket$, let $X^n[\mathcal{S}]$ denote the components of $X^{n}$ whose indices are in $\mathcal{S}$. Let $\oplus$ denote addition modulo two. In the following, all logarithms are taken base two. We also give a notation overview in Table \ref{tab:notation}.

\begin{table*}[h!]
\centering
\caption{Notation overview.} 
\begin{tabular}{|l|l|}
\hline
\textbf{Notation} & \textbf{Description} \\ \hline
$\llbracket a,b \rrbracket$ & Integer interval $\llbracket a,b\rrbracket \triangleq [ \lfloor a \rfloor, \lceil b \rceil]\cap \mathbb{N}$ \\ 
$\bar{z}$ & For $z \in \{1,2\}$, $\bar{z} \triangleq 3-z$ \\ 
$\oplus$ & Addition modulo 2  \\ 
$n$ & Number of channel uses \\ 
$(X_i)_t$ & Channel input of Server $i \in \{1,2\}$ at time $t \in \llbracket 1,n \rrbracket$ \\ 
$Y_t$ & Output of the binary adder MAC at time $t \in \llbracket 1,n \rrbracket$, $Y_t \triangleq (X_1)_t + (X_2)_t$ \\ 
$Y^t$ & Sequence of outputs up to time $t \in \llbracket 1,n \rrbracket$, $Y^t \triangleq (Y_i)_{i\in \llbracket 1 , t\rrbracket}$ \\ 
$X_i^n[\mathcal{S}]$ & Components of vector $X_i^n \triangleq ((X_i)_1, \dots, (X_i)_n)$ indexed by $\mathcal{S} \subseteq \llbracket 1,n \rrbracket$, $i\in\{1,2\}$ \\ 
$L_i$ & Number of files stored at Server $i \in \{1,2\}$ \\ 
$\mathcal{L}_i$ & Set of file indices for Server $i\in\{1,2\}$, $\mathcal{L}_i \triangleq \llbracket 1, L_i \rrbracket$, $i\in\{1,2\}$ \\ 
$F_{i,l}$ & File $l \in \mathcal{L}_i$ stored at Server $i\in\{1,2\}$, uniformly distributed over $\{0,1\}^{nR_i}$ \\ 
$U_0,U_1,U_2$ & Local randomness available at the client, Server 1, and Server 2, respectively \\ 
$Z_1,Z_2$ & Client's file selection for Server 1 and Server 2, respectively  \\ 
$\mathbf{A}$ & Total public communication \\ 
$\hat{F}_{i,Z_i}$ & Client's estimate of the requested file $F_{i,Z_i}$, $i\in\{1,2\}$ \\ 
\hline
\end{tabular}
\label{tab:notation}

\end{table*}

Consider a binary adder multiple-access channel (MAC) $(\mathcal{Y}, p_{Y|X_1X_2}, \mathcal{X}_1 \times \mathcal{X}_2)$ defined by 
\begin{align} \label{eqchannel}
Y \triangleq X_1 + X_2,
\end{align}
where  $Y \in \mathcal{Y} \triangleq \{ 0,1,2\}$ is the output and  $X_1$, $X_2\in  \mathcal{X}_1 \triangleq  \mathcal{X}_2 \triangleq \{ 0,1\}$ are the inputs. In the following, we assume that all the participants in the protocol strictly follow the protocol. The setting is summarized in Figure \ref{fig2} and formally described in the following definitions.

\begin{figure}
\centering
  \includegraphics[width=7.75 cm]{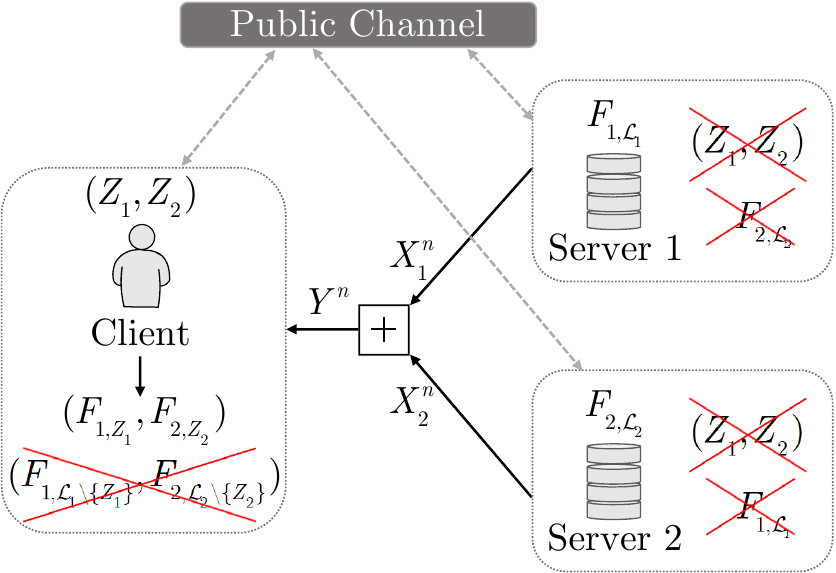}
  \caption{Dual-source symmetric private information retrieval between one client and two servers who have access to a public channel and a binary adder MAC as described in \eqref{eqchannel}. The file selection of the client is $(Z_1,Z_2)\in \mathcal{L}_1 \times \mathcal{L}_2$, i.e., the client wishes to obtain $(F_{1,Z_1}, F_{2,Z_2})$ from the servers with the constraints that the client must not learn information about $(F_{1,\mathcal{L}_1 \backslash\{{Z}_1\}}, F_{2,\mathcal{L}_2 \backslash\{{Z}_2\}})$,  $(Z_1,Z_2)$ must remain private from both servers, and Server $1$ must not learn information about Server $2$'s content, and vice versa for Server $2$.}
  \label{fig2}
\end{figure}
\begin{defn}
An $(n,L_1,L_2,2^{nR_1},2^{nR_2})$ dual-source SPIR protocol consists of
\begin{itemize}
\item Two non-colluding servers and one client;
\item For $i\in\{1,2\}$, $L_i$ independent random variables $(F_{i,l})_{l\in\mathcal{L}_i}$ uniformly distributed over $\{ 0,1\}^{nR_i}$, where $\mathcal{L}_i \triangleq \llbracket 1, L_i \rrbracket$, which represent $L_i$ files stored at Server~$i$; 
\item $U_0$, $U_1$, $U_2$, three independent random variables, which represent local randomness available at the client, Server~$1$, and Server~$2$, respectively;
\item $Z_1$, $Z_2$, two independent random variables uniformly distributed over $\mathcal{L}_1$ and $\mathcal{L}_2$, respectively, which represent the client file choice for Servers $1$, and $2$, i.e., $(Z_1,Z_2) =(i,j)$ means that the client is requesting the files $(F_{1,i},F_{2,j})$, where $(i,j) \in \mathcal{L}_1 \times \mathcal{L}_2$;
\end{itemize}
and operates as follows, 
\begin{itemize}
\item  For time $t=1$ to time $t =n$, the servers send $((X_1)_t,(X_2)_t) \in \mathcal{X}_1 \times \mathcal{X}_2$ over the binary adder MAC \eqref{eqchannel} and the client observes $Y_t \triangleq (X_1)_t+(X_2)_t$, where $(X_1)_t$ and $(X_2)_t$ are functions of $  (F_{1,\mathcal{L}_1},U_1)$ and $   (F_{2,\mathcal{L}_2},U_2)$, respectively. Define $Y^n \triangleq (Y_t)_{t\in \llbracket 1 , n\rrbracket}$.
\item Next, the servers and the client are allowed to communicate, possibly interactively in $r$ rounds of communication, over a noiseless channel. Specifically, for $j \in \llbracket 1 , r \rrbracket$, Servers~1 and~2 form and publicly send the messages $M_{1}(j)$ and $M_{2}(j)$, respectively, which are functions of $(F_{1,\mathcal{L}_1},U_1, (M_{0,1}(m))_{m \in \llbracket 1, j-1 \rrbracket })$ and $(F_{2,\mathcal{L}_2},U_2,(M_{0,2}(m))_{m \in \llbracket 1, j-1 \rrbracket })$, respectively.  The client forms and publicly sends the messages $M_{0,1}(j)$ and $M_{0,2}(j)$, which are functions of $(Z_1, U_0, Y^n, (M_{1}(m))_{m \in \llbracket 1, j \rrbracket })$ and $( Z_2,U_0, Y^n, (M_{2}(m))_{m \in \llbracket 1, j \rrbracket })$, respectively. 
\end{itemize} 
 Define  $\mathbf{A}_1 \triangleq (M_{0,1}(j),M_{1}(j))_{j \in \llbracket 1, r \rrbracket}$ and $\mathbf{A}_2 \triangleq (M_{0,2}(j),M_{2}(j))_{j \in \llbracket 1, r \rrbracket}$. The entire public communication is denoted by $\mathbf{A} \triangleq (\mathbf{A}_1,\mathbf{A}_2)$.
Finally, the client forms $\hat{F}_{1,Z_1}$, an estimate of $F_{1,Z_1}$, from $(Y^n,U_0,\mathbf{A}_1,Z_1)$, and $\hat{F}_{2,Z_2}$, an estimate of $F_{2,Z_2}$, from $(Y^n,U_0,\mathbf{A}_2,Z_2)$. 
\end{defn}

\begin{defn} \label{defmal}
A rate pair $(R_1, R_2)$ is achievable if there exists a sequence of $(n,L_1,L_2,2^{nR_1},2^{nR_2})$ dual-source SPIR protocols such that  
\begin{align}
&\lim_{n\to \infty} \mathbb{P}\left[(\widehat{F}_{1,Z_1},\widehat{F}_{2,Z_2}) \neq (F_{1,Z_1},F_{2,Z_2}) \right]  = 0 \label{eqre}\\    &\hspace*{12em} \text{(File recovery at Client)}, \nonumber \\
&\lim_{n\to \infty} I\left( F_{1, \mathcal{L}_1} X_1^n U_1 \mathbf{A} ; Z_1 Z_2\right)   = 0 \label{eqpra} \\ & \hspace*{12em}\text{(Client privacy w.r.t. Server 1)}, \nonumber \\
&\lim_{n\to \infty} I\left( F_{2, \mathcal{L}_2} X_2^n U_2\mathbf{A} ; Z_1 Z_2\right)   = 0  \label{eqprb} \\  &\hspace*{12em}\text{(Client privacy w.r.t. Server 2)},\nonumber \\
&\lim_{n\to \infty}I(F_{1, \mathcal{L}_1} X^n_{1 } U_1 \mathbf{A} ;F_{2, \mathcal{L}_{2}}  ) = 0 \label{eqps12} \\ &\hspace*{11em}\text{(Server 2 privacy w.r.t. Server 1)}, \nonumber\\
&\lim_{n\to \infty}I(F_{2, \mathcal{L}_2} X^n_{2 } U_2 \mathbf{A} ;F_{1, \mathcal{L}_{1}}  ) = 0 \label{eqps21}\\  &\hspace*{11em}\text{(Server 1 privacy w.r.t. Server 2)}, \nonumber \\
&\lim_{n\to \infty} I\left(  Z_1 Z_2 Y^n U_0 \mathbf{A}; F_{1, \mathcal{L}_1 \backslash\{{Z}_1\}} F_{2, \mathcal{L}_2 \backslash\{{Z}_2\}} \right)    = 0 \label{eqpr2}\\ &\hspace*{12em}\text{(Servers' privacy w.r.t. Client)}.  \nonumber
\end{align}
The set of all achievable rate pairs is called the dual-source SPIR  capacity region   and is denoted by~$\mathcal{C}_{\textup{SPIR}^2}(L_1,L_2)$.
\end{defn}

Equation~\eqref{eqre} ensures that the client obtains the selected files. Equations \eqref{eqpra} and \eqref{eqprb} ensure the client's privacy by keeping the file selection $(Z_1,Z_2)$  private from Server $1$ and Server~$2$, respectively. Note that $(F_{i, \mathcal{L}_i}, X_i^n, U_i, \mathbf{A})$ represents all the information available at Server $i\in \{1,2\}$ at the end of the protocol. Equation \eqref{eqps12} (respectively Equation \eqref{eqps21}) ensures Server $2$'s (respectively Server $1$'s) privacy with respect to Server $1$'s (respectively Server $2$'s). Equation~\eqref{eqpr2} ensures the servers' privacy by keeping all the non-selected files private from the client.  Note that $( Z_1, Z_2, Y^n, U_0 ,\mathbf{A})$ represents all the information available to the client at the end of the protocol.

Note that, if the servers are colluding, then no positive rates are achievable, as the setting reduces to an SPIR setting between one client and one server over a noiseless channel, for which it is known that information-theoretic security cannot be achieved, e.g., \cite{sun2018capacity}.

\section{Main Results} \label{secres}
\begin{figure}
\centering
  \includegraphics[width=7.1 cm]{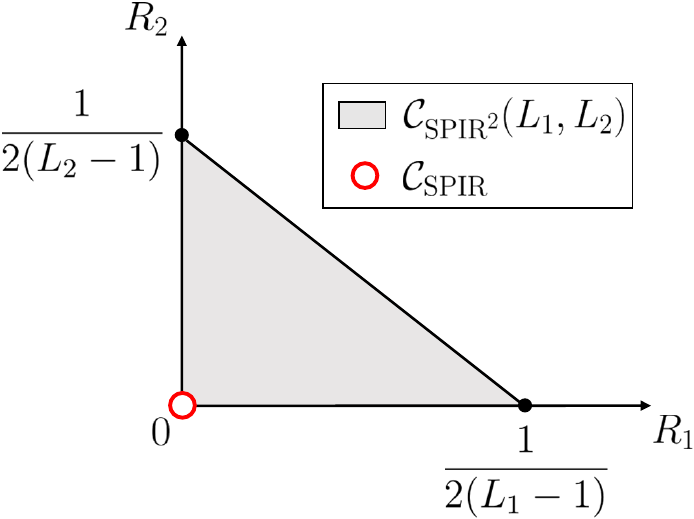}
  \caption{ The x-coordinate (resp. y-coordinate) of $\mathcal{C}_{\textup{SPIR}}$ is defined as the SPIR capacity between Server $1$ (resp. Server~$2$) and client in the absence of Server $2$ (resp. Server $1$), i.e.,  $\mathcal{C}_{\textup{SPIR}}=(0,0)$. $\mathcal{C}_{\textup{SPIR}^2}(L_1,L_2)$ is  the dual-source SPIR  capacity region characterized in Theorem \ref{th1}.}
  \label{fig1}
\end{figure}
Our main results  provide a full characterization of  the dual-source SPIR  capacity region.
\begin{thm} \label{th2}
The dual-source SPIR capacity region is \begin{multline*}
    \mathcal{C}_{\textup{SPIR}^2}(L_1,L_2) \\= \left\{ (R_1,R_2) : (L_1-1) R_1 + (L_2-1) R_2 \leq  \frac{1}{2} \right\}.\end{multline*}
Moreover, any rate pair in $\mathcal{C}_{\textup{SPIR}^2}(L_1,L_2)$ is achievable without time-sharing.
\end{thm}
\begin{proof}
The converse and achievability are proved in Sections~\ref{secproof2} and \ref{secproof3}, respectively.
\end{proof}

\begin{cor}
If Servers $1$ and $2$ have the same number of files, i.e., $L_1=L=L_2$, then $$\mathcal{C}_{\textup{SPIR}^2}(L) = \left\{ (R_1,R_2) :  R_1 +  R_2 \leq  \frac{1}{2(L-1)} \right\}.$$
Moreover, if $L_1=L=L_2$ and the files on both servers have the same size, i.e.,  $nR_1 = nR=nR_2$, then $$\mathcal{C}_{\textup{SPIR}^2}(L) = \left\{ (R,R) :  R \leq  \frac{1}{4(L-1)} \right\}.$$
\end{cor}

The results contrast with the case where the client performs independent SPIR with each of the two servers, i.e., with one server at a time and without any intervention of the other server, since in this case information-theoretic SPIR is impossible over a noiseless communication channel without any additional resources. We illustrate this point in Figure \ref{fig1}.

\section{Achievability part of Theorem \ref{th2}} \label{secproof3}
We first treat the special case $L_1=L_2=2$ in Section \ref{secL1L22}. Then, we  show in Section~\ref{secL1L2} that the case $L_1=L_2$ can be reduced to the special case $L_1=L_2=2$. Finally, we treat the general case $L_1\neq L_2$ in Section \ref{secgeneral} and show that it can also be reduced to the case $L_1=L_2=2$.

\subsection{Special case $L_1 = L_2=2$} \label{secL1L22}
Suppose that $L_1 = L_2=2$. Our coding scheme, based on \cite{ahlswede2013oblivious}, is described in Algorithm \ref{alg:ot}. The idea of Algorithm \ref{alg:ot} is to exploit the additive channel structure $Y^n = X_1^n + X_2^n$. Positions where $Y_t \in \{ 0,2\}$ (``good'' indices) let the client recover the channel inputs, while positions with $Y_t = 1$ (``bad'' indices) reveal nothing about the inputs. Using this partition, the client constructs two sets of indices $(S_1^{(i)}, S_2^{(i)})$ for each Server $i\in\{1,2\}$, one corresponding to recoverable positions, the other to unrecoverable ones, without revealing which is which. Each server then applies a one-time pad to its inputs over both sets as follows.  On the set the client can recover, the masked input allows reconstruction of the desired file. On the set the client cannot recover, the masked input hides information about the non-requested file. This ensures the client obtains exactly the requested file while preserving both client and server privacy. Algorithm \ref{alg:ot} is formally  analyzed as follows.  
\begin{algorithm}
  \caption{Dual-source SPIR when $L_1=L_2=2$}
  \label{alg:ot}
  \begin{algorithmic}[1]
  \vskip 6pt
  \Require $t<1/2$ and  $\alpha\in [0,1]$. \vskip 6pt
  \State  The channel \eqref{eqchannel} is used as follows: \vskip 3pt
  \begin{algsubstates}
        \State Consider $(X_1^n,X_2^n)$ distributed according to the uniform distribution over $\{ 0,1 \}^{2n}$ 
        \vskip 3pt
        \State Servers $1$ and $2$ send $X_1^n$ and $X_2^n$, respectively, over the channel \eqref{eqchannel} \vskip 3pt
        \State The client observes $Y^n \triangleq X_1^n + X_2^n $
      \end{algsubstates}  \vskip 6pt

\State Upon observation of $Y^n$, the client \vskip 3pt

\begin{algsubstates}
        \State Defines \vspace{-.7em}
 	\begin{align*}
    \mathcal{G} &\triangleq \{ i \in \llbracket 1 , n \rrbracket : Y_i \in \{ 0,2\} \}, \\	\mathcal{B} &\triangleq \{ i \in \llbracket 1 , n \rrbracket : Y_i =1 \} ;\end{align*} 

        \State Defines $M \triangleq \min (|\mathcal{G}| , |\mathcal{B}|)$; \vskip 3pt
        
        \State Constructs $\mathcal{G}_1, \mathcal{G}_2$ such that $\mathcal{G}_1 \cup \mathcal{G}_2 \subset \mathcal{G} $, $\mathcal{G}_1 \cap \mathcal{G}_2 = \emptyset $, and $(|\mathcal{G}_1|,|\mathcal{G}_2|) = (\alpha M, \bar{\alpha} M)$; \vskip 3pt
        
        \State Constructs $\mathcal{B}_1, \mathcal{B}_2$ such that $\mathcal{B}_1 \cup \mathcal{B}_2 \subset \mathcal{B} $, $\mathcal{B}_1 \cap \mathcal{B}_2 = \emptyset $, and $(|\mathcal{B}_1|,|\mathcal{B}_2|) = (\alpha M, \bar{\alpha} M)$;\footnotemark \vskip 3pt
        
        \State Defines for $i \in \{ 1,2 \}$ \begin{align*}
	S_1^{(i)} \triangleq \begin{cases} 
	\mathcal{G}_i & \text{if }Z_i = 1 \\
	\mathcal{B}_i & \text{if }Z_i = 2 \\
	\end{cases},\quad
		S_2^{(i)} \triangleq \begin{cases} 
	\mathcal{B}_i & \text{if }Z_i = 1 \\
	\mathcal{G}_i & \text{if }Z_i = 2 \\
	\end{cases}.
\end{align*}
Note that when $Z_i = j$, $i\in \{ 1,2\}$, $j \in \{ 1,2\}$, the client can determine $X_i^n[S_j^{(i)}]$.
      \end{algsubstates} \vskip 6pt

\State  If 
\begin{align} \label{eqsize}
\left|  \frac{|\mathcal{G}|}{n} - \frac{1}{2}   \right| \leq n^{-t},
\end{align}
then the client sends to Servers $1$ and $2$ the sets $S_1^{(1)},S_2^{(1)},S_1^{(2)},S_2^{(2)}$, otherwise  the client aborts the protocol.\vskip 6pt
\State The public communication of the servers is as follows: \vskip 3pt
\begin{algsubstates}
        \State Server $1$ sends to the client $(M_{11}, M_{12})$ where
\begin{align*}
(M_{11}, M_{12})& \triangleq (X_1^n[S_1^{(1)}] \oplus F_{1,1}, X_1^n[S_2^{(1)}] \oplus F_{1,2}). 
\end{align*}
\vskip -3pt
\State Server $2$ sends to the client $(M_{21}, M_{22})$ where
\begin{align*}
(M_{21}, M_{22})& \triangleq (X_2^n[S_1^{(2)}] \oplus F_{2,1}, X_2^n[S_2^{(2)}] \oplus F_{2,2}).
\end{align*}
       \end{algsubstates}
\State The client obtains its file selection as follows:\vskip 3pt
\begin{algsubstates}
        \State  If $Z_1 =i\in \{ 1,2\}$, then the client determines $X_1^n[S_i^{(1)}]$ and computes $M_{1i} \oplus  X_1^n[S_i^{(1)}] = F_{1,i}$.\vskip 3pt\State If $Z_2 =i\in \{ 1,2\}$, then the client determines $X_2^n[S_i^{(2)}]$ and computes $M_{2i} \oplus  X_2^n[S_i^{(2)}] = F_{2,i}$.
       \end{algsubstates}
  \end{algorithmic}
\end{algorithm}
\subsubsection{Reliability}
By construction, \eqref{eqre} is true unless $\left| \frac{|\mathcal{G}|}{n} - \frac{1}{2}  \right| > n^{-t}$, which happens with negligible probability $O(n^{-1+2t})$ by Chebyshev's inequality. 
\subsubsection{Privacy of the receiver selection}
The entire public communication between the receiver and the senders   is $\mathbf{A} \triangleq (M_{11}, M_{12},M_{21}, M_{22}, S_1^{(1)},S_2^{(1)},S_1^{(2)},S_2^{(2)})$. We prove that Server 1 does not learn the  file choices of the client as follows 
\begin{align}
	& I(F_{1,1}F_{1,2}X_1^n \mathbf{A};Z_1Z_2) \nonumber \\ \nonumber
	& \stackrel{(a)} =  I(F_{1,1}F_{1,2}X_1^n  M_{21} M_{22}  S_1^{(1)}S_2^{(1)}S_1^{(2)}S_2^{(2)};Z_1Z_2) \\ \nonumber
		& \stackrel{(b)} \leq  I(F_{1,1}F_{1,2}X_1^n    S_1^{(1)}S_2^{(1)}S_1^{(2)}S_2^{(2)};Z_1Z_2) \\ \nonumber
		& \stackrel{(c)} =  I(X_1^n    S_1^{(1)}S_2^{(1)}S_1^{(2)}S_2^{(2)};Z_1Z_2) \\ \nonumber
		& \stackrel{(d)} = I(X_1^n    S_1^{(1)}S_2^{(1)};Z_1) + I(X_1^n    S_1^{(1)}S_2^{(1)};Z_2|Z_1) \\ \nonumber
		& \phantom{--} + I(S_1^{(2)}S_2^{(2)};Z_2|X_1^n    S_1^{(1)}S_2^{(1)}) \\ \nonumber
		& \phantom{--} + I(S_1^{(2)}S_2^{(2)};Z_1|Z_2 X_1^n    S_1^{(1)}S_2^{(1)}) \nonumber\\ \nonumber
	& \stackrel{(e)} \leq  I(X_1^n[\mathcal{G}_1 \cup \mathcal{B}_1]    S_1^{(1)}S_2^{(1)};Z_1)\\ \nonumber
		& \phantom{--}  + I(X_1^n[\mathcal{G}_2 \cup \mathcal{B}_2]    S_1^{(2)}S_2^{(2)};Z_2) \\ 
	& = 0, \label{eqa1}
	\end{align}\footnotetext{As an example, suppose $n=12$, $\alpha = \scriptstyle\frac{1}{3}$, and $Y^n = (1, 0, 2, 0, 1, 2, 0, 1, 1, 2, 1, 1 )$. Then, $M=6$, $\mathcal{G} = \{ 2 , 3 ,4, 6 ,7,10\}$,  $\mathcal{B} = \{1,5,8,9,11,12 \}$, and, for instance, $\mathcal{G}_1 = \{ 2 , 3 \}$, $\mathcal{G}_2 = \{4, 6 ,7,10\}$, $\mathcal{B}_1 = \{1,5 \}$, $\mathcal{B}_2 = \{8,9,11,12 \}$.} 
where
\begin{enumerate}[(a)]
    \item  holds because $(M_{11},M_{12})$ is a function of $(X_1^n, F_{1,1},F_{1,2},S_1^{(1)},S_2^{(1)} )$; \item holds by~\eqref{eqotp1} in Lemma~\ref{lemotp} (Appendix \ref{applemotp}) with $\epsilon = \delta =0$ and the substitutions $A \leftarrow (Z_1, Z_2)$, $B \leftarrow (F_{1,1}F_{1,2},X_1^n,S_1^{(1)},S_2^{(1)}, S_1^{(2)},S_2^{(2)}  )$, $C \leftarrow (X_2^n[S_1^{(2)}] , X_2^n[S_2^{(2)}] )$, $D \leftarrow (F_{2,1}F_{2,2})$; 
    \item holds because $(F_{1,1}F_{1,2})$ is independent of $(X_1^n,S_1^{(1)},S_2^{(1)},S_1^{(2)},S_2^{(2)},Z_1,Z_2)$; 
    \item holds by the chain rule; 
    \item holds because $I(X_1^n    S_1^{(1)}S_2^{(1)};Z_2|Z_1) =0$ by independence between $Z_2$ and $(X_1^n ,   S_1^{(1)},S_2^{(1)},Z_1)$, \begin{align*} 
    &I(S_1^{(2)}S_2^{(2)};Z_1|Z_2 X_1^n    S_1^{(1)}S_2^{(1)}) \\
    &\leq I(S_1^{(2)}S_2^{(2)} Z_2 X_1^n[\mathcal{G}_2 \cup \mathcal{B}_2];Z_1 X_1^n[\mathcal{G}_1 \cup \mathcal{B}_1]    S_1^{(1)}S_2^{(1)}) \\
    &=0,\end{align*} and we also have  \begin{align*} 
    & I(S_1^{(2)}S_2^{(2)};Z_2|X_1^n    S_1^{(1)}S_2^{(1)}) \\ &\leq   I(S_1^{(2)}S_2^{(2)} X_1^n[\mathcal{G}_2 \cup \mathcal{B}_2];Z_2|X_1^n[\mathcal{G}_1 \cup \mathcal{B}_1] S_1^{(1)}S_2^{(1)}) \\
    & =I(S_1^{(2)}S_2^{(2)} X_1^n[\mathcal{G}_2 \cup \mathcal{B}_2];Z_2).\end{align*}
\end{enumerate}

Next, by exchanging the role of the servers in the proof of \eqref{eqa1}, we obtain that Server 2 does not learn the file choices of the client, i.e.,
\begin{align*}
	 I(F_{2,1}F_{2,2}X_2^n \mathbf{A};Z_1Z_2) = 0. 
	\end{align*}

\subsubsection{Privacy of the unselected strings for the senders}
We now prove that the client does not learn any information about the unselected files in \eqref{eqh1}
\begin{figure*}
\begin{align*}
	 I\left(  Z_1 Z_2 Y^n \mathbf{A}; F_{1,\bar{Z}_1} F_{2,\bar{Z}_2} \right)
	& \stackrel{(a)} = I\left(  Z_1 Z_2 Y^n  M_{11} M_{12} M_{21} M_{22}  ; F_{1,\bar{Z}_1} F_{2,\bar{Z}_2} \right) \\
	& \stackrel{(b)} = \frac{1}{4}\sum_{z_1,z_2} I\left(  Y^n  M_{11} M_{12} M_{21} M_{22}  ; F_{1,\bar{z}_1} F_{2,\bar{z}_2} |  Z_1 =z_1, Z_2 = z_2 \right) \\
	& \stackrel{(c)} = \frac{1}{4}\sum_{z_1,z_2} I\left(  Y^n F_{1,z_1} F_{2,z_2} ,X_1^n[S^{(1)}_{\bar{z}_1}] \oplus F_{1,\bar{z}_1},  X_2^n[S^{(2)}_{\bar{z}_2}] \oplus F_{2,\bar{z}_2} ; F_{1,\bar{z}_1} F_{2,\bar{z}_2} |  Z_1 =z_1, Z_2 = z_2  \right) \\
	& \stackrel{(d)} = \frac{1}{4}\sum_{z_1,z_2} I\left(  Y^n ,X_1^n[S^{(1)}_{\bar{z}_1}] \oplus F_{1,\bar{z}_1},  X_2^n[S^{(2)}_{\bar{z}_2}] \oplus F_{2,\bar{z}_2} ; F_{1,\bar{z}_1} F_{2,\bar{z}_2} |  Z_1 =z_1, Z_2 = z_2  \right) \\
	& \stackrel{(e)} = 0 \numberthis \label{eqh1}
\end{align*}
\hrulefill
\end{figure*}
where 
\begin{enumerate}[(a)]
    \item holds because $(S_1^{(1)},S_2^{(1)},S_1^{(2)},S_2^{(2)})$ is a function of $(Y^n,Z_1,Z_2)$;
    \item holds because $(Z_1,Z_2)$ is independent of $(F_{1,\bar{Z}_1} ,F_{2,\bar{Z}_2})$;
    \item holds because there is a one-to-one mapping between $(Y^n,  M_{11}, M_{12})$ and $(Y^n,F_{1,z_1}  ,X_1^n[S^{(1)}_{\bar{z}_1}] \oplus F_{1,\bar{z}_1})$, and between $(Y^n, M_{21}, M_{22} )$ and $(Y^n,F_{2,z_2} , X_2^n[S^{(2)}_{\bar{z}_2}] \oplus F_{2,\bar{z}_2})$;
    \item holds because $( F_{1,z_1}, F_{2,z_2} )$ is independent of $(Y^n ,X_1^n[S^{(1)}_{\bar{z}_1}] \oplus F_{1,\bar{z}_1},  X_2^n[S^{(2)}_{\bar{z}_2}] \oplus F_{2,\bar{z}_2} , F_{1,\bar{z}_1} ,F_{2,\bar{z}_2})$;
    \item holds by \eqref{eqotp2} in Lemma \ref{lemotp} (Appendix \ref{applemotp}) with $\epsilon = \delta = 0$ and the substitutions $A \leftarrow Y^n$, $C \leftarrow (F_{1,\bar{z}_1}, F_{2,\bar{z}_2})$, $D \leftarrow (X_1^n[S^{(1)}_{\bar{z}_1}] ,X_2^n[S^{(2)}_{\bar{z}_2}] ) $. 
\end{enumerate}

 \subsubsection{Servers’ mutual privacy} 
 We have 
\begin{align}
&I\left(  F_{1,1} F_{1,2} X_1^n  \mathbf{A}; F_{2,1} F_{2,2} \right) \nonumber \\
& \stackrel{(a)} =  I(F_{1,1} F_{1,2} X_1^n  M_{21} M_{22}  S_1^{(1)}S_2^{(1)}S_1^{(2)}S_2^{(2)};F_{2,1} F_{2,2}) \nonumber \\ \nonumber
& \stackrel{(b)} = 0,
\end{align}
where 
\begin{enumerate}[(a)]
    \item holds because $(M_{11},M_{12})$ is a function of $(X_1^n, F_{1,1} ,F_{1,2},S_1^{(1)},S_2^{(1)} )$;
    \item holds by \eqref{eqotp2} in  Lemma~\ref{lemotp} with the substitutions  $\epsilon \leftarrow 0 $, $\delta \leftarrow 0 $, $A \leftarrow (F_{1,1}, F_{1,2},X_1^n,S_1^{(1)},S_2^{(1)},S_1^{(2)},S_2^{(2)})$, $C \leftarrow ( F_{2,1}, F_{2,2} )$, $D \leftarrow (X_2^n[S_1^{(2)}],X_2^n[S_2^{(2)}])$.
\end{enumerate}
By exchanging the roles of the servers we also have \begin{align*}
I\left(  F_{2,1} F_{2,2} X_2^n \mathbf{A}; F_{1,1} F_{1,2} \right) 
 = 0.
\end{align*}

\subsubsection{Achieved rate pair} \label{secrate2}
The rate pair $(R_1 = \frac{\alpha}{2}, R_2 = \frac{\bar{\alpha}}{2})$ is achieved since we have 
\begin{align*}
	R_1 & =  \frac{|\mathcal{G}_1|}{n}=\alpha \frac{M}{n} \xrightarrow{n \to \infty} \frac{\alpha}{2},\\
		R_2 & = \frac{|\mathcal{G}_2|}{n}= \bar{\alpha} \frac{M}{n}\xrightarrow{n \to \infty} \frac{\bar{\alpha}}{2},
\end{align*}
where the limits hold because we have $M = \frac{1}{2} \left( |\mathcal{G}| + |\mathcal{B}|  - \big||\mathcal{G}| - |\mathcal{B}| \big| \right) $, and $\frac{|\mathcal{G}|}{n}\xrightarrow{n \to \infty} \frac{1}{2}$  and $\frac{|\mathcal{B}|}{n} = 1- \frac{|\mathcal{G}|}{n} \xrightarrow{n \to \infty} \frac{1}{2}$ by \eqref{eqsize}.

\subsection{Case $L_1=L_2$} \label{secL1L2}
In this section, we present an achievability scheme for the case $L_1 = L_2 = L$ separately, as it is significantly simpler than the achievability scheme for the general case $L_1 \neq L_2$, which is discussed in Section \ref{secgeneral}. The idea is to construct a coding scheme for this case by utilizing $L-1$ times Algorithm~\ref{algrtoc1} developed for the case $L_1=L_2=2$. This reduction idea is inspired by~\cite{brassard1996oblivious} in the context of oblivious transfer. Before we describe our coding scheme in Algorithm~\ref{algrtoc1}, we illustrate the coding scheme idea in the following example. 
\begin{ex} \label{ex1}
    Suppose that $L_1 =L_2 = 3$, and $S$ and $T$ are uniformly distributed over $\llbracket 1,2^{nR_1} \rrbracket $ and $\llbracket 1,2^{nR_2} \rrbracket $, respectively. Servers 1 and 2 engage in a dual-source SPIR protocol from Section~\ref{secL1L22}, interacting with the client twice.  In the first interaction, the client selects between \( S \) and \( F_{1,1} \) from Server 1, and between \( T \) and \( F_{2,1} \) from Server 2. During the second interaction, the client chooses between \( F_{1,2} \oplus S \) and \( F_{1,3} \oplus S \) from Server~1, and between \( F_{2,2} \oplus T \) and \( F_{2,3} \oplus T \) from Server 2.  Specifically, in the first round, if the client desires \( F_{1,1} \), they request \(F_{1,1} \) from Server~1, otherwise, they request \( S \). In the second round, the client requests \( F_{1,2} \oplus S \) or \( F_{1,3} \oplus S \) from Server~1, depending on whether they wish to retrieve \( F_{1,2} \) or \( F_{1,3} \), respectively. The same logic applies to Server 2. More specifically, the series of requests made by the client to retrieve the desired files is summarized in Table \ref{tab2}.
    \begin{table*}[!h]
\caption{Series of requests for file selection at each server in Example \ref{ex1}.} 
    \centering
  \renewcommand{\arraystretch}{1.3}  
\begin{tabular}{l c c c | c c c} \centering
 &  & Server 1 & & & Server 2 & \\
\hline
Request 1 & $F_{1,1}$ & $S$ & $S$ & $F_{2,1}$ & $T$ & $T$ \\
Request 2 & $F_{1,2}\oplus S$ & $F_{1,2}\oplus S$ & $F_{1,3}\oplus S$ & $F_{2,2}\oplus T$ & $F_{2,2}\oplus T$ & $F_{2,3}\oplus T$ \\
\hline
\textbf{Recover} & $F_{1,1}$ & $F_{1,2}$ & $F_{1,3}$ & $F_{2,1}$ & $F_{2,2}$ & $F_{2,3}$ 
\end{tabular} \label{tab2}
\end{table*}
\end{ex}
Algorithm \ref{algrtoc1} extends the idea of Example \ref{ex1} to $L_1 = L_2 = L$, using $L-1$ sequential interactions between the server and the client. At Server $i \in \{1,2\}$, each file, except the first, is hidden with random masks $(S_{i,t})$, for $t \in \llbracket 1, L-2 \rrbracket$. In each round $t \in \llbracket 1, L-1 \rrbracket$, the client selects one of two  values $(C_{i,t}[1], C_{i,t}[2])$, which are defined as some XOR combinations of masks and files, from Server $i$ via Algorithm~\ref{alg:ot}. Across the $L-1$ rounds, the client carefully chooses which values to recover, accumulates them, and combines them with XOR operations to reconstruct the desired file $F_{i,Z_i}$ from Server $i$.

\begin{algorithm}
  \caption{Dual-source SPIR when $L_1=L_2=L$}
  \label{algrtoc1}
  \begin{algorithmic}[1]
  \vskip 6pt
  \Require $L-2$ sequences $(S_{1,t})_{t \in \llbracket 1,L-2 \rrbracket  }$ uniformly distributed over $\{ 0,1\}^{nR_1}$, $L-2$ sequences $(S_{2,t})_{t \in \llbracket 1,L-2 \rrbracket  }$ uniformly distributed over $\{ 0,1\}^{nR_2}$, the file selection $(Z_1,Z_2) \in \mathcal{L}_1 \times \mathcal{L}_2$ \vskip 6pt
\State  Server $j\in \llbracket 1,2 \rrbracket$ forms $(C_{j,t})_{t \in \llbracket 1,L-1 \rrbracket  }$ as follows:
\begin{align*}
(C_{j,1}[1],C_{j,1}[2])  &\!\triangleq \! (F_{j,1}, S_{j,1}) \\
(C_{j,t}[1],C_{j,t}[2]) &\!\triangleq \! (F_{j,t} \oplus S_{j,t-1},S_{j,t-1} \oplus S_{j,t}) \\
(C_{j,L-1}[1],C_{j,L-1}[2]) &\!\triangleq \! (F_{j,L-1} \!\oplus \! S_{j,L-2},S_{j,L-2} \!\oplus \! F_{j,L})
\end{align*}
where $t  \in \llbracket 2 , L-2 \rrbracket$.

\noindent{}For $t  \in \llbracket 1 , L-1 \rrbracket$, define $C_{j,t} \triangleq (C_{j,t}[1],C_{j,t}[2])$.
  \vskip 6pt
\State The client forms $(Z_{j,t})_{j \in \llbracket 1 ,2 \rrbracket,t \in \llbracket 1, L-1 \rrbracket}$ as follows:
\begin{align*}
Z_{j,t} & \triangleq  1+\mathds{1} \{ t < Z_j \}   , \forall t \in \llbracket 1, L-1 \rrbracket, \forall j \in \llbracket 1 ,2 \rrbracket 
\end{align*}
\For{ $t  \in \llbracket 1 , L-1 \rrbracket$}  
\State The client and the two servers perform the SPIR protocol in Algorithm \ref{alg:ot}  with the two sequences $(C_{j,t}[1],C_{j,t}[2])$ at Server $j \in \llbracket 1,2 \rrbracket$ and the selection $(Z_{1,t},Z_{2,t})$ for the client. 
\EndFor
\State By Lines 2-5, the client can form for $j \in\llbracket 1 ,2 \rrbracket $
\begin{align*}
F_{j,Z_j} =  \begin{cases} 
C_{j,Z_j}[1] \oplus \bigoplus_{t=1}^{Z_j-1} C_{j,t}[2] & \text{ if } Z_j < L\\
 \bigoplus_{t=1}^{Z_j} C_{j,t}[2] & \text{ if } Z_j = L
\end{cases}
\end{align*}
  \end{algorithmic}
\end{algorithm}

From Section \ref{secrate2}, the lengths of the files in Algorithm \ref{alg:ot} for the case $L_1=L_2=2$ and $n$ channel uses is
$\left(\alpha M, \bar{\alpha} M \right)$,
hence, the rate pair achieved with Algorithm \ref{algrtoc1}, which utilizes $L-1$ times  Algorithm \ref{alg:ot} is 
\begin{align*}
	R_1 & =  \frac{\alpha M}{n(L-1)}\xrightarrow{n \to \infty} \frac{\alpha}{2(L-1)},\\
		R_2 & =  \frac{\bar{\alpha} M}{n(L-1)}\xrightarrow{n \to \infty} \frac{\bar{\alpha}}{2(L-1)}.
\end{align*}
The privacy analysis of Algorithm \ref{algrtoc1} is similar to that of Algorithm \ref{algrtoc}, presented in the next section, and is therefore~omitted.

 \subsection{General case} \label{secgeneral}
 Consider arbitrary values for $L_1$ and $L_2$. The idea is to construct a coding scheme for this case by utilizing  $(L_1-1) \times (L_2-1)$ times  Algorithm \ref{alg:ot}, developed for the case $L_1=L_2=2$, with subsequences of files from Server 1 and Server 2, with size $nR_1/(L_2-1)$ and $nR_2/(L_1-1)$, respectively.

Before we describe our coding scheme in Algorithm \ref{algrtoc}, we illustrate the coding scheme idea in the following example.
\begin{ex} \label{ex2}
 Suppose that $L_1 =3$, $L_2 = 4$.  $S_1,S_2,S_3$ are uniformly distributed over $\llbracket 1,2^{nR_1/3} \rrbracket $ and $T_1,T_2,T_3,T_4$ are uniformly distributed over $\llbracket 1,2^{nR_2/2} \rrbracket $, respectively. Server 1 breaks down their files in three parts, i.e., for $i\in \{1,2,3\}$, $F_{1,i}= (F_{1,i,1},F_{1,i,2},F_{1,i,3})$, and Server 2 breaks down their files in two parts, i.e., for $j\in \{1,2,3,4\}$, $F_{2,j}= (F_{2,j,1},F_{2,j,2})$. Servers 1 and 2 engage in a dual-source SPIR protocol from Section~\ref{secL1L22}, interacting with the client $(L_1-1) \times (L_2-1)= 6$ times such that the client obtains one element of each of the six pairs in each of the two  three-by-two matrices defined in \eqref{eqmat2}.
 \begin{figure*}
 \begin{align}
 M_1 \triangleq \begin{pmatrix}
(F_{1,1,1},S_1) & (F_{1,2,1}\oplus S_1,F_{1,3,1}\oplus S_1) \\
(F_{1,1,2},S_2) & (F_{1,2,2}\oplus S_2,F_{1,3,2}\oplus S_2) \\
(F_{1,1,3},S_3) & (F_{1,2,3}\oplus S_3,F_{1,3,3}\oplus S_3)
\end{pmatrix}\!, \quad
M_2 \triangleq \begin{pmatrix}
(F_{2,1,1},T_1) & (F_{2,1,2},T_3) \\
(F_{2,2,1}\oplus T_1,T_1\oplus T_2) & (F_{2,2,2}\oplus T_3,T_3\oplus T_4) \\
(F_{2,3,1}\oplus T_2,F_{2,4,1}\oplus T_2) & (F_{2,3,2}\oplus T_4,F_{2,4,2}\oplus T_4)
\end{pmatrix} \label{eqmat2}
\end{align}
\hrulefill
\end{figure*}
Observe that the client can recover part $k \in \{1,2,3 \}$ of File $i\in\{1,2,3\}$ from Server 1 with their choices of elements in row $k$ of $M_1$. Similarly, the client can recover part $l \in \{1,2\}$ of File $j\in\{1,2,3,4\}$ from Server 2 with their choices of elements in column $l$ of $M_2$. More specifically, the series of requests made by the client to retrieve the desired files is summarized in Table \ref{tableIII}.\\
 \begin{table*} \centering \caption{Series of requests for file selection at each server in Example \ref{ex2}.} 
 \renewcommand{\arraystretch}{1.3}  
{\small{ \begin{tabular}{l c c c | c c c c}
 &  & Server 1 & & & Server 2 & \\
\hline
Request 1 & $F_{1,1,1}$ & $S_1$ & $S_1$ & $F_{2,1,1}$ & $T_1$ & $T_1$ & $T_1$\\
Request 2 & $F_{1,2,1} \oplus S_1$ & $F_{1,2,1}\oplus S_1$ & $F_{1,3,1}\oplus S_1$ & $F_{2,1,2}$ & $T_3$ & $T_3$ & $T_3$\\
Request 3 & $F_{1,1,2}$ & $S_2$ & $S_2$ & $F_{2,2,1}\oplus T_1$ & $F_{2,2,1}\oplus T_1$ & $T_1\oplus T_2$ & $T_1\oplus T_2$ \\
Request 4 & $F_{1,2,2} \oplus S_2$ &  $F_{1,2,2}\oplus S_2$ & $F_{1,3,2}\oplus S_2$ & $F_{2,2,2}\oplus T_3$ & $F_{2,2,2}\oplus T_3$ & $T_3\oplus T_4$  & $T_3\oplus T_4$\\
Request 5 & $F_{1,1,3}$ & $S_3$ & $S_3$ & $F_{2,3,1}\oplus T_2$ & $F_{2,3,1}\oplus T_2$ & $F_{2,3,1}\oplus T_2$  & $F_{2,4,1}\oplus T_2$\\
Request 6 & $F_{1,2,3} \oplus S_3$ & $F_{1,2,3}\oplus S_3$ & $F_{1,3,3}\oplus S_3$ & $F_{2,3,2}\oplus T_4$ & $F_{2,3,2}\oplus T_4$ & $F_{2,3,2}\oplus T_4$ & $F_{2,4,2}\oplus T_4$\\
\hline
\textbf{Recover} & $F_{1,1}$ & $F_{1,2}$ & $F_{1,3}$ & $F_{2,1}$ & $F_{2,2}$ & $F_{2,3}$ &$F_{2,4}$
\end{tabular}}}\label{tableIII}\end{table*}
\end{ex}
 Algorithm \ref{algrtoc} extends the idea of Example \ref{ex2} to arbitrary $L_1$ and  $L_2 $, using $(L_1-1) \times (L_2-1)$ sequential interactions between the server and the client. Server $i \in \{1,2\}$ breaks down each file in $L_{\bar{i}} -1$ parts, and each part, except the first, is hidden with random masks. Then, similar to Algorithm \ref{algrtoc1}, in each round $k \in \llbracket 1, (L_1-1)\times (L_2-1) \rrbracket$, the client selects one of two masked values from Server $i$ via Algorithm \ref{alg:ot}. Across the $(L_1-1)\times (L_2-1)$ rounds, the client carefully chooses which values to recover, accumulates them, and combines them with XOR operations to reconstruct the $L_{\bar{i}} -1$ parts of the desired file $F_{i,Z_i}$ from Server $i$.

\begin{algorithm*}
  \caption{Dual-source SPIR when $L_1\neq L_2$}
  \label{algrtoc}
  \begin{algorithmic}[1]
  \vskip 6pt
  \Require $(L_1-2)\times(L_2-1)$ sequences $(S_{1,t,i})_{t \in \llbracket 1,L_1-2 \rrbracket, i\in \llbracket 1,L_2-1 \rrbracket  }$ uniformly distributed over $\{ 0,1\}^{nR_1/(L_2-1)}$, $(L_2-2) \times (L_1-1)$ sequences $(S_{2,t,j})_{t \in \llbracket 1,L_2-2 \rrbracket, j\in \llbracket 1,L_1-1 \rrbracket  }$ uniformly distributed over $\{ 0,1\}^{nR_2/(L_1-1)}$, the file selection $(Z_1,Z_2) \in \mathcal{L}_1 \times \mathcal{L}_2$ \vskip 6pt
\State  Server $1$ breaks down File $F_{1,l}$, $l\in\mathcal{L}_1$, in $L_2-1$ parts denoted by $(F_{1,l,i})_{i \in \llbracket 1, L_2-1 \rrbracket}$   \vskip 6pt
\State  Server $2$ breaks down File $F_{2,l}$, $l\in\mathcal{L}_2$, in $L_1-1$ parts denoted by $(F_{2,l,j})_{j \in \llbracket 1, L_1-1 \rrbracket}$   \vskip 6pt
\State Server $1$ forms $(C_{1,t,i})_{t \in \llbracket 1,L_1-1 \rrbracket ,i \in \llbracket 1,L_2-1 \rrbracket }$ as follows: for $i \in \llbracket 1,L_2-1 \rrbracket$,
\begin{align*}
(C_{1,1,i}[1],C_{1,1,i}[2]) & \triangleq (F_{1,1,i}, S_{1,1,i}), \\
(C_{1,t,i}[1],C_{1,t,i}[2]) &\triangleq (F_{1,t,i} \oplus S_{1,t-1,i},S_{1,t-1,i} \oplus S_{1,t,i}), \forall t  \! \in \! \llbracket 2 , L_1 \! -2 \rrbracket, \\
(C_{1,L_1-1,i}[1],C_{1,L_1-1,i}[2])
&\triangleq (F_{1,L_1-1,i} \oplus S_{1,L_1-2,i},S_{1,L_1-2,i} \oplus F_{1,L_1,i}),
\end{align*}
and for $t  \in \llbracket 1 , L_1-1 \rrbracket$, define $C_{1,t,i} \triangleq (C_{1,t,i}[1],C_{1,t,i}[2]).$ Consider also the following notation  $$(\bar{C}_{1,k})_{k \in \llbracket 1 , (L_1-1)(L_2-1)\rrbracket} \triangleq (C_{1,t,i})_{t \in \llbracket 1,L_1-1 \rrbracket ,i \in \llbracket 1,L_2-1 \rrbracket }.$$
  \vskip 6pt
  \State Server $2$ forms $(C_{2,t,i})_{t \in \llbracket 1,L_2-1 \rrbracket ,i \in \llbracket 1,L_1-1 \rrbracket }$ as follows: for $i \in \llbracket 1,L_1-1 \rrbracket$,
\begin{align*}
(C_{2,1,i}[1],C_{2,1,i}[2]) & \triangleq (F_{2,1,i}, S_{2,1,i}), \\
(C_{2,t,i}[1],C_{2,t,i}[2]) 
& \triangleq  (F_{2,t,i} \oplus S_{2,t-1,i},S_{2,t-1,i} \oplus S_{2,t,i}), \forall t \! \in \! \llbracket 2 , L_2-2 \rrbracket, \\
(C_{2,L_2-1,i}[1],C_{2,L_2-1,i}[2]) &\triangleq (F_{2,L_2-1,i} \oplus S_{2,L_2-2,i},S_{2,L_2-2,i} \oplus F_{2,L_2,i}),
\end{align*}
and for $t  \in \llbracket 1 , L_2-1 \rrbracket$, define $C_{2,t,i} \triangleq (C_{2,t,i}[1],C_{2,t,i}[2])$. Consider also the following notation  $$(\bar{C}_{2,k})_{k \in \llbracket 1 , (L_1-1)(L_2-1)\rrbracket} \triangleq (C_{2,t,i})_{t \in \llbracket 1,L_2-1 \rrbracket ,i \in \llbracket 1,L_1-1 \rrbracket }.$$ 
  \vskip 6pt
\State The client forms $(Z_{1,t,i})_{t \in \llbracket 1, L_1-1 \rrbracket, i \in \llbracket 1 ,L_2-1 \rrbracket}$ and $(Z_{2,t,i})_{t \in \llbracket 1, L_2-1 \rrbracket, i \in \llbracket 1 ,L_1-1 \rrbracket}$ as follows:
\begin{align*}
Z_{1,t,i} & \triangleq  1+\mathds{1} \{ t < Z_1 \}  , \forall t \! \in \! \llbracket 1, L_1 \! -1 \rrbracket, \forall i \! \in \! \llbracket 1 ,L_2 \! -1 \rrbracket \\
Z_{2,t,j} & \triangleq  1+\mathds{1} \{ t < Z_2 \}  , \forall t   \! \in \! \llbracket 1, L_2 \! -1 \rrbracket, \forall j \! \in \! \llbracket 1 ,L_1 \! -1 \rrbracket 
\end{align*}
Consider also the following notation  
\begin{align*}
(Z_{1,k})_{k \in \llbracket 1, (L_1-1)(L_2-1)\rrbracket}  \triangleq  (Z_{1,t,i})_{t \in \llbracket 1, L_1-1 \rrbracket, i \in \llbracket 1 ,L_2-1 \rrbracket},\quad 
(Z_{2,k})_{k \in \llbracket 1, (L_1-1)(L_2-1)\rrbracket}  \triangleq  (Z_{2,t,j})_{t \in \llbracket 1, L_2-1 \rrbracket, j \in \llbracket 1 ,L_1-1 \rrbracket}
\end{align*}
\For{ $k  \in \llbracket 1 , (L_1-1)(L_2-1) \rrbracket$}  
\State The client and the two servers perform the SPIR protocol in Algorithm \ref{alg:ot} with the two sequences $(\bar{C}_{j,k}[1],\bar{C}_{j,k}[2])$ at Server $j \in \llbracket 1,2 \rrbracket$ and the selection $(Z_{1,k},Z_{2,k})$ for the client.

The subscript $k$ is used in the notation of the  random variables $(X_{1,k}^{n},X_{2,k}^{n},Y_{k}^{n}, \mathbf{A}_k, Z_{1,k}, Z_{2,k})$ involved in this SPIR protocol.

\EndFor
\State By Lines 3-5, the client can reconstruct the $L_2 -1$ parts of $F_{1,Z_1}$ and the $L_1 -1$ parts of $F_{2,Z_2}$ as follows:
\begin{align*}
F_{1,Z_1,i} =  \begin{cases} 
C_{1,Z_1,i}[1] \oplus \bigoplus_{t=1}^{Z_1-1} C_{1,t,i}[2] & \text{ if } Z_1 < L_1\\
 \bigoplus_{t=1}^{Z_1} C_{1,t,i}[2] & \text{ if } Z_1 = L_1
\end{cases}, \forall i \in \llbracket 1, L_2 -1 \rrbracket,
\end{align*}
and
\begin{align*}
F_{2,Z_2,j} =  \begin{cases} 
C_{2,Z_2,j}[1] \oplus \bigoplus_{t=1}^{Z_2-1} C_{2,t,j}[2] & \text{ if } Z_2 < L_2\\
 \bigoplus_{t=1}^{Z_2} C_{2,t,j}[2] & \text{ if } Z_2 = L_2
\end{cases}, \forall j \in \llbracket 1, L_1-1 \rrbracket
\end{align*}
  \end{algorithmic}
\end{algorithm*}

\subsubsection{Server's privacy with respect to  Client}
Define 
\begin{align*}
K &\triangleq (L_1-1)(L_2-1) \\
 W_t &\triangleq (\bar{C}_{1,t} [Z_{1,t}], \bar{C}_{2,t} [Z_{2,t}]), \forall t\in \llbracket 1, K\rrbracket, \\ 
W_{1:K} & \triangleq (\bar{C}_{1,t} [Z_{1,t}], \bar{C}_{2,t} [Z_{2,t}])_{t \in \llbracket 1 , K \rrbracket},\\  
Y^n_{1:K} & \triangleq (Y^n_t )_{t \in \llbracket 1 , K \rrbracket}, \\
\mathbf{A}_{1:K} & \triangleq (\mathbf{A}_t )_{t \in \llbracket 1 , K \rrbracket}, \\   (Z_{1,1:K},Z_{2,1:K}) & \triangleq (Z_{1,t},Z_{2,t})_{t \in \llbracket 1 , K \rrbracket} ).\end{align*}

We will need the following two lemmas.

\begin{lem} \label{lem1a}
 We have   \begin{align*}
 I(Z_{1} Z_{2} Y^{n}_{1:K} \mathbf{A}_{1:K}; \bar{C}_{1,1:K}   \bar{C}_{2,1:K}   | W_{1:K} ) 
 = 0. \numberthis \label{eq1a2}
\end{align*}
\end{lem}
\begin{proof}
    See Appendix \ref{App_lem1}.
\end{proof}

\begin{lem} \label{lem2a}
 We have   
 \begin{align*} 
 I(W_{1:K} ; F_{1, \mathcal{L}_1}F_{2, \mathcal{L}_2} | F_{1,Z_1} F_{2,Z_2} )  = 0.  \numberthis \label{eq1b2}
\end{align*}
\end{lem}
\begin{proof}
    See Appendix \ref{App_lem2}.
\end{proof}
Then, using Lemma \ref{lem1a} and Lemma \ref{lem2a}, we prove server's privacy with respect to the client as follows. We have
\begin{align*}
&I(Z_{1} Z_2 Y^{n}_{1:K} \mathbf{A}_{1:K}; F_{1, \mathcal{L}_1\backslash \{ Z_1\}}F_{2, \mathcal{L}_2\backslash \{ Z_2\}}  ) \\
& = I(Z_{1} Z_2 Y^{n}_{1:K} \mathbf{A}_{1:K}; F_{1, \mathcal{L}_1}F_{2, \mathcal{L}_2} | F_{1,Z_1} F_{2,Z_2} ) \\
& \leq I(Z_{1} Z_2 Y^{n}_{1:K} \mathbf{A}_{1:K} W_{1:K}; F_{1, \mathcal{L}_1}F_{2, \mathcal{L}_2} | F_{1,Z_1} F_{2,Z_2} ) \\
& = I( W_{1:K}; F_{1, \mathcal{L}_1}F_{2, \mathcal{L}_2} | F_{1,Z_1} F_{2,Z_2} ) \\
& \phantom{--}+ I(Z_{1} Z_2 Y^{n}_{1:K} \mathbf{A}_{1:K} ; F_{1, \mathcal{L}_1}F_{2, \mathcal{L}_2} | F_{1,Z_1} F_{2,Z_2}  W_{1:K}) \\
& \stackrel{(a)}=  I(Z_{1} Z_2 Y^{n}_{1:K} \mathbf{A}_{1:K} ; F_{1, \mathcal{L}_1}F_{2, \mathcal{L}_2} |  W_{1:K}) \\
& \leq  I(Z_{1} Z_2 Y^{n}_{1:K} \mathbf{A}_{1:K} ; F_{1, \mathcal{L}_1}F_{2, \mathcal{L}_2} \bar{C}_{1,1:K}   \bar{C}_{2,1:K}|  W_{1:K}) \\
& =  I(Z_{1} Z_2 Y^{n}_{1:K} \mathbf{A}_{1:K} ;  \bar{C}_{1,1:K}   \bar{C}_{2,1:K}|  W_{1:K})\\
& \phantom{--} +  I(Z_{1} Z_2 Y^{n}_{1:K} \mathbf{A}_{1:K} ; F_{1, \mathcal{L}_1}F_{2, \mathcal{L}_2} |\bar{C}_{1,1:K}   \bar{C}_{2,1:K} W_{1:K})\\
& \stackrel{(b)}= I(Z_{1} Z_2 Y^{n}_{1:K} \mathbf{A}_{1:K} ;  \bar{C}_{1,1:K}   \bar{C}_{2,1:K}|  W_{1:K}) \\
& \stackrel{(c)}= 0,
\end{align*}
where 
\begin{enumerate} [(a)]
    \item holds by \eqref{eq1b2} and because $(F_{1,Z_1}, F_{2,Z_2})$ can be recovered from $W_{1:K}$; 
    \item holds because $(Z_{1}, Z_2, Y^{n}_{1:K}, \mathbf{A}_{1:K}) - ( \bar{C}_{1,1:K} ,\bar{C}_{2,1:K},  W_{1:K}) - (F_{1, \mathcal{L}_1},F_{2, \mathcal{L}_2} )$ forms a Markov~chain; 
    \item holds by \eqref{eq1a2}.
\end{enumerate}
\subsubsection{Client privacy with respect to the servers}
Fix $j \in \llbracket 1 , 2 \rrbracket$. We have  
\begin{align*}
& I(F_{j, \mathcal{L}_j} X_{j,1:K}^n \mathbf{A}_{1:K};Z_1Z_2)\\
&  \stackrel{(a)}= I(F'_{j, 1:K} X_{j,1:K}^n \mathbf{A}_{1:K};Z_1Z_2)\\
& = \sum_{t = 1}^{K} I(F'_{j, t} X_{j,t}^n \mathbf{A}_{t};Z_1Z_2|F'_{j, 1:t-1} X_{j,1:t-1}^n \mathbf{A}_{1:t-1})\\
& \stackrel{(b)}= \sum_{t = 1}^{K} I(F'_{j, t} X_{j,t}^n \mathbf{A}_{t};Z_1Z_2 Z_{1,t}Z_{2,t}|F'_{j, 1:t-1} X_{j,1:t-1}^n \mathbf{A}_{1:t-1})\\
& \stackrel{(c)} = \sum_{t = 1}^{K} I(F'_{j, t} X_{j,t}^n \mathbf{A}_{t};  Z_{1,t}Z_{2,t}|F'_{j, 1:t-1} X_{j,1:t-1}^n \mathbf{A}_{1:t-1})\\
& \stackrel{(d)}\leq  \sum_{t = 1}^{K} I(S^{\star}  F'_{j, t} X_{j,t}^n \mathbf{A}_{t};  Z_{1,t}Z_{2,t}|F'_{j, 1:t-1} X_{j,1:t-1}^n \mathbf{A}_{1:t-1})\\
& \stackrel{(e)}= \sum_{t = 1}^{K}  I(  F'_{j, t}  X_{j,t}^n \mathbf{A}_{t};  Z_{1,t}Z_{2,t}|S^{\star} F'_{j, 1:t-1} X_{j,1:t-1}^n \mathbf{A}_{1:t-1}) \\
& \stackrel{(f)} \leq \sum_{t = 1}^{K} I(  F'_{j, t}  X_{j,t}^n \mathbf{A}_{t} ;  Z_{1,t}Z_{2,t}| S^{\star}  ) \\
&\leq \sum_{t = 1}^{K} I(S^{\star} \bar{C}_{j,t}[1]\bar{C}_{j,t}[2]  F'_{j, t}  X_{j,t}^n \mathbf{A}_{t} ;  Z_{1,t}Z_{2,t}   ) \\
& \stackrel{(g)} = \sum_{t = 1}^{K} I(S^{\star}  \bar{C}_{j,t}[1]\bar{C}_{j,t}[2]     X_{j,t}^n \mathbf{A}_{t} ;  Z_{1,t}Z_{2,t}   ) \\
& \stackrel{(h)} = \sum_{t = 1}^{K} I(S^{\star}  ;  Z_{1,t}Z_{2,t}  | \bar{C}_{j,t}[1]\bar{C}_{j,t}[2]    X_{j,t}^n \mathbf{A}_{t} ) \\
&\stackrel{(i)} = 0,
\end{align*}
where 
\begin{enumerate} [(a)]
    \item holds  for $k \in \llbracket 1, L_j-1 \rrbracket$, $l\in \llbracket1 ,L_{\bar{j}}-1 \rrbracket$, with the definition $$F'_{j, k,l} \triangleq \begin{cases} F_{j, t,l} & \text{ if } k \in \llbracket 1, L_j-2 \rrbracket\\
(F_{j,L_j-1,l},F_{j,L_j,l}) & \text{ if } k = L_j-1\end{cases},$$ and $(F'_{j, t} )_{t \in \llbracket 1, K \rrbracket}\triangleq(F'_{j, k,l} )_{k \in \llbracket 1, L_j-1 \rrbracket, l \in \llbracket1 ,L_{\bar{j}}-1 \rrbracket}$;
\item holds because $(Z_{1,t},Z_{2,t})$ is a function of $(Z_1,Z_2)$ for any $t \in \llbracket 1 , K \rrbracket$; 
\item holds because $ (F'_{j, t}, X_{j,t}^n, \mathbf{A}_{t}) - (F'_{j, 1:t-1}, X_{j,1:t-1}^n ,\mathbf{A}_{1:t-1},Z_{1,t},Z_{2,t}) - (Z_1,Z_2)$ forms a Markov chain; 
\item holds with the definition $$S^{\star}\triangleq \begin{cases} 
F'_{j,t}\oplus \bar{C}_{j,t}[1] & \text{if } |F'_{j,t}|=1\\
F'_{j,t}\oplus \bar{C}_{j,t} & \text{if } |F'_{j,t}|=2\end{cases};$$
\item  holds by the chain rule and because for any $t \in \llbracket 1 , K \rrbracket$
\begin{align*}
&I(S^{\star} ;  Z_{1,t}Z_{2,t}|F'_{j, 1:t-1} X_{j,1:t-1}^n \mathbf{A}_{1:t-1})\\
& \leq I(S^{\star} ; \bar{C}_{j,t-1} Z_{1,t}Z_{2,t} F'_{j, 1:t-1} X_{j,1:t-1}^n \mathbf{A}_{1:t-1})\\
& = I(S^{\star};  \bar{C}_{j,t-1}) \\
&\phantom{-}+ I(S^{\star} ; Z_{1,t}Z_{2,t} F'_{j, 1:t-1} X_{j,1:t-1}^n \mathbf{A}_{1:t-1}|  \bar{C}_{j,t-1})\\
& = 0,
\end{align*}
where the last equality holds by the one-time pad lemma and because $S^{\star} - \bar{C}_{j,t-1}- ( Z_{1,t},Z_{2,t} ,F'_{j, 1:t-1}, X_{j,1:t-1}^n, \mathbf{A}_{1:t-1})  $ forms a Markov chain; 
\item holds because for any $t \in \llbracket 1 , K \rrbracket$
\begin{align*}
&I(  F'_{j, t}  X_{j,t}^n \mathbf{A}_{t};  Z_{1,t}Z_{2,t}|S^{\star} F'_{j, 1:t-1} X_{j,1:t-1}^n \mathbf{A}_{1:t-1})\\
& \leq H(  F'_{j, t}  X_{j,t}^n \mathbf{A}_{t} |S^{\star})\\
&\phantom{-}- H(  F'_{j, t}  X_{j,t}^n \mathbf{A}_{t} |  Z_{1,t}Z_{2,t} S^{\star} F'_{j, 1:t-1} X_{j,1:t-1}^n \mathbf{A}_{1:t-1}) \\
&  = H(  F'_{j, t}  X_{j,t}^n \mathbf{A}_{t} |S^{\star})- H(  F'_{j, t}  X_{j,t}^n \mathbf{A}_{t} |  Z_{1,t}Z_{2,t} S^{\star}  ) \\
& = I(  F'_{j, t}  X_{j,t}^n \mathbf{A}_{t} ;  Z_{1,t}Z_{2,t}| S^{\star}  )
\end{align*}
where the first inequality holds by the definition of the mutual information and because conditioning reduces entropy, and the first equality  holds because $ (F'_{j, t} , X_{j,t}^n, \mathbf{A}_{t}  )-  (Z_{1,t},Z_{2,t}, S^{\star}) - (F'_{j, 1:t-1}, X_{j,1:t-1}^n, \mathbf{A}_{1:t-1})   $ forms a Markov chain; 
\item holds because $F'_{j, t}$ can be recovered from $(\bar{C}_{j,t}[1],\bar{C}_{j,t}[2],S^{\star})$  for any $t \in \llbracket 1 , K \rrbracket$; 
\item holds  by the chain rule and because for any $t \in \llbracket 1 , K \rrbracket$, $  I(  \bar{C}_{j,t}[1]\bar{C}_{j,t}[2]     X_{j,t}^n \mathbf{A}_{t} ;  Z_{1,t}Z_{2,t}   ) = 0$ by the privacy of the client selection ensured in Line 7 of Algorithm \ref{algrtoc}; 
\item holds because for any $t \in \llbracket 1 , L-1 \rrbracket$
\begin{align*}
&I(S^{\star} ;  Z_{1,t}Z_{2,t}  | \bar{C}_{j,t}[1]\bar{C}_{j,t}[2]    X_{j,t}^n \mathbf{A}_{t} ) \\
& \leq    H(S^{\star}) - H(S^{\star} |  Z_{1,t}Z_{2,t}   \bar{C}_{j,t}[1] \bar{C}_{j,t}[2]    X_{j,t}^n \mathbf{A}_{t} ) \\
& =   H(S^{\star}) - H(S^{\star} |      \bar{C}_{j,t}[1]\bar{C}_{j,t}[2]  ) \\
& =   I(S^{\star} ;   \bar{C}_{j,t}[1]\bar{C}_{j,t}[2] ) \\
& = 0,
\end{align*}
where  the first inequality holds by the definition of the mutual information and because conditioning reduces entropy, the first equality holds because $S^{\star}-(\bar{C}_{j,t}[1],\bar{C}_{j,t}[2])-(Z_{1,t},Z_{2,t}   , X_{j,t}^n, \mathbf{A}_{t})$ forms a Markov chain, and the last equality holds by the one-time pad lemma.
\end{enumerate}
\subsubsection{Servers' mutual privacy}
We first show the following lemma.
\begin{lem} \label{lem3a}
    We have
    \begin{align*}
  I( \bar{C}_{1,1:K}  X^n_{1,1:K} \mathbf{A}_{1:K};\bar{C}_{2,1:K})=  0. \numberthis \label{eqlim2m}
\end{align*}
\end{lem}
\begin{proof}
    See Appendix \ref{App_lem3}.
\end{proof}
Then, using Lemma \ref{lem3a}, we have
\begin{align*}
 &I(F_{1, \mathcal{L}_1} X^n_{1,1:K} \mathbf{A}_{1:K};F_{2, \mathcal{L}_{2}}  ) \\
 & \leq  I( \bar{C}_{1,1:K} F_{1, \mathcal{L}_1} X^n_{1,1:K} \mathbf{A}_{1:K};F_{2, \mathcal{L}_{2}}  ) \\
 &  = I( \bar{C}_{1,1:K}  X^n_{1,1:K} \mathbf{A}_{1:K};F_{2, \mathcal{L}_{2}}  ) \\
&\phantom{-}+ I( F_{1, \mathcal{L}_1};F_{2, \mathcal{L}_{2}}  | \bar{C}_{1,1:K}  X^n_{1,1:K} \mathbf{A}_{1:K} )\\
  & \stackrel{(a)} = I( \bar{C}_{1,1:K}  X^n_{1,1:K} \mathbf{A}_{1:K};F_{2, \mathcal{L}_{2}}  )  \\
  &  \leq I( \bar{C}_{1,1:K}  X^n_{1,1:K} \mathbf{A}_{1:K};\bar{C}_{2,1:K}  F_{2, \mathcal{L}_{2}}  )  \\
    & = I( \bar{C}_{1,1:K}  X^n_{1,1:K} \mathbf{A}_{1:K};\bar{C}_{2,1:K}   ) \\
&\phantom{-} + I( \bar{C}_{1,1:K}  X^n_{1,1:K} \mathbf{A}_{1:K};  F_{2, \mathcal{L}_{2}}  | \bar{C}_{2,1:K})\\
        & \stackrel{(b)}= I( \bar{C}_{1,1:K}  X^n_{1,1:K} \mathbf{A}_{1:K};\bar{C}_{2,1:K}   )  \\
        & \stackrel{(c)}= 0,
\end{align*} 
where \begin{enumerate}[(a)]
    \item holds because \begin{align*}
    &I( F_{1, \mathcal{L}_1};F_{2, \mathcal{L}_{2}}  | \bar{C}_{1,1:K}  X^n_{1,1:K} \mathbf{A}_{1:K} ) \\
    &\leq I( F_{1, \mathcal{L}_1}\bar{C}_{1,1:K}  X^n_{1,1:K};F_{2, \mathcal{L}_{2}}  |  \mathbf{A}_{1:K} )\\
    &=0;\end{align*} 
    \item holds because \begin{align*}
    &I( \bar{C}_{1,1:K}  X^n_{1,1:K} \mathbf{A}_{1:K};  F_{2, \mathcal{L}_{2}}  | \bar{C}_{2,1:K}) \\
    &= I( \mathbf{A}_{1:K};  F_{2, \mathcal{L}_{2}}  | \bar{C}_{2,1:K}) \\
    &\phantom{-} I( \bar{C}_{1,1:K}  X^n_{1,1:K} ;  F_{2, \mathcal{L}_{2}}  | \bar{C}_{2,1:K} \mathbf{A}_{1:K}) \\
    &\leq I( \mathbf{A}_{1:K};  F_{2, \mathcal{L}_{2}}  | \bar{C}_{2,1:K}) \\
    &\phantom{-}+ I( \bar{C}_{1,1:K}  X^n_{1,1:K} ;  F_{2, \mathcal{L}_{2}}   \bar{C}_{2,1:K} |\mathbf{A}_{1:K})\\
    &= I( \mathbf{A}_{1:K};  F_{2, \mathcal{L}_{2}}  | \bar{C}_{2,1:K})=0;\end{align*} 
    \item  holds by~\eqref{eqlim2m}.
\end{enumerate}

Similarly, we have
\begin{align*}
 I(F_{2, \mathcal{L}_2} X^n_{2,1:K} \mathbf{A}_{1:K};F_{1, \mathcal{L}_{1}}  ) = 0.
\end{align*} 
\subsubsection{Rates}

From Section \ref{secrate2}, the lengths of the files in Algorithm \ref{alg:ot} for the case $L_1=L_2=2$ and $n$ channel uses is
$\left(\alpha M, \bar{\alpha} M \right)$,
hence, the rate pair achieved with Algorithm \ref{algrtoc1}, which utilizes $(L_1-1) \times (L_2-1)$ times  Algorithm \ref{alg:ot} with subsequences of files from Server 1 and Server 2, with size $nR_1/(L_2-1)$ and $nR_2/(L_1-1)$, respectively, is 
\begin{align*}
	R_1 & =  (L_2-1) \frac{ \alpha M}{n(L_1-1)  (L_2-1)}\xrightarrow{n \to \infty} \frac{\alpha}{2(L_1-1)},\\
		R_2 & = (L_1-1) \frac{\bar{\alpha} M}{n(L_1-1)  (L_2-1)}\xrightarrow{n \to \infty} \frac{\bar{\alpha}}{2(L_2-1)}.
\end{align*}
\begin{remark}   The download cost from Server~$i\in\{1,2\}$ is $2nR_i(L_i-1)$ bits, i.e.,  $2(L_i-1)$ bits per bit of downloaded file. For Server~$i\in\{1,2\}$, the total number of bitwise XOR operations is $2n R_i (2L_i-3)$, i.e., $2 (2L_i-3)$ bitwise XOR operations per bit of downloaded file.
 \end{remark}
\section{Converse part of Theorem \ref{th2}} \label{secproof2}
We first establish the following outer bound on the dual-source SPIR capacity region.
\begin{prop} \label{th1}
The dual-source SPIR capacity region is such that \begin{align*} 
\mathcal{C}_{\textup{SPIR}^2}(L_1,L_2)  & \subseteq \left\{ (R_1,R_2) :(L_1-1) R_1 + (L_2-1) R_2 \right.\\
& \phantom{------l}\left. \leq  \textstyle\max_{p_{X_1}p_{X_2}}H(X_1X_2|Y)\right\}.\end{align*}
\end{prop}
\begin{proof}

Consider a sequence of $(n,L_1,L_2,2^{nR_1},2^{nR_2})$ dual-source SPIR protocols that achieve the rate pair $(R_1,R_2)$. We prove the outer bound through a series of lemma. 
\begin{lem} \label{lema}
For any $z_1, z_2 \in \mathcal{L}_1 \times \mathcal{L}_2$, we have 
\begin{align}
	 I(F_{1,z_1}; Y^nU_0 | X_1^n  \mathbf{A}_1, Z_1=z_1, Z_2 =z_2 ) & =0, \label{eqlema1}\\
	 	 I(F_{2,z_2}; Y^nU_0 | X_2^n  \mathbf{A}_2, Z_1=z_1, Z_2 =z_2 ) & =0. \label{eqlema2}
\end{align}	
\end{lem}
\begin{proof}
It is sufficient to prove \eqref{eqlema1}, as the proof of \eqref{eqlema2} can be obtained by exchanging the roles of the servers.  We have \eqref{eqM}
\begin{figure*}
\begin{align}
	 I(F_{1, \mathcal{L}_1} U_1; Y^n U_0 | X_1^n  \mathbf{A}_1    Z_1 Z_2) \label{eqrep1} 
	&  = I(F_{1, \mathcal{L}_1}  U_1 ; Y^n U_0 | X_1^n  (M_{0,1}(j),M_{1}(j))_{j \in \llbracket 1, r \rrbracket}  Z_1 Z_2)  \\ \nonumber
	& \leq  I(F_{1, \mathcal{L}_1} U_1; Y^n U_0 M_{0,1}(r)| X_1^n  (M_{0,1}(j),M_{1}(j))_{j \in \llbracket 1, r -1 \rrbracket}M_{1}(r)   Z_1 Z_2)\\ \nonumber
	& \stackrel{(a)}=  I(F_{1, \mathcal{L}_1} U_1; Y^n U_0  | X_1^n  (M_{0,1}(j),M_{1}(j))_{j \in \llbracket 1, r -1 \rrbracket}M_{1}(r)   Z_1 Z_2)\\ \nonumber
	& \leq  I(F_{1, \mathcal{L}_1} U_1 M_{1}(r) ; Y^n U_0  | X_1^n  (M_{0,1}(j),M_{1}(j))_{j \in \llbracket 1, r -1 \rrbracket}   Z_1 Z_2)\\ 
	& \stackrel{(b)}=  I(F_{1, \mathcal{L}_1} U_1   ; Y^n U_0  | X_1^n  (M_{0,1}(j),M_{1}(j))_{j \in \llbracket 1, r -1 \rrbracket}   Z_1 Z_2) \label{eqrep2} \\ \nonumber
	& \stackrel{(c)} \leq   I(F_{1, \mathcal{L}_1} U_1   ; Y^n U_0  | X_1^n    Z_1 Z_2)\\ \nonumber
	& \stackrel{(d)} = I(F_{1, \mathcal{L}_1} U_1;   U_0 | X_1^n   Z_1 Z_2) \\ \nonumber
	& \leq I(F_{1, \mathcal{L}_1} U_1  X_1^n  ;   U_0 |     Z_1 Z_2) \\  \nonumber
	& \stackrel{(e)}= I(F_{1, \mathcal{L}_1} U_1;  U_0 |      Z_1 Z_2) \\
    & = 0 \label{eqM}
\end{align}
\hrulefill
\end{figure*}
where 
\begin{enumerate}[(a)]
    \item holds by the chain rule and because by definition $M_{0,1}(r)$ is a function of $(Z_1,U_0,Y^n,(M_{1}(j))_{j \in \llbracket 1, r  \rrbracket})$;
    \item holds by the chain rule and because $M_{1}(r)$ is a function of $(F_{1, \mathcal{L}_1}, U_1 ,(M_{0,1}(j))_{j \in \llbracket 1, r-1  \rrbracket}) $;
    \item holds by repeating $r-1$ times the steps between Equations \eqref{eqrep1} and \eqref{eqrep2};
    \item holds because $ (F_{1, \mathcal{L}_1},U_1)-  (U_0, X_1^n,    Z_1, Z_2) - Y^n$ forms a Markov chain;
    \item holds by the chain rule and because $X_1^n$ is a function of $(F_{1, \mathcal{L}_1},U_1 )$. 
    \end{enumerate}
    Next, for any $z_1, z_2 \in  \mathcal{L}_1 \times \mathcal{L}_2$, we have 
\begin{align}
	& 	I(F_{1,z_1}; Y^n U_0 | X_1^n  \mathbf{A}_1, Z_1=z_1, Z_2 =z_2 ) \nonumber \\ \nonumber
	& \leq I(F_{1, \mathcal{L}_1} U_1; Y^n U_0 | X_1^n \mathbf{A}_1, Z_1 =z_1, Z_2=z_2)\\ \nonumber
	& \leq (\mathbb{P}[(Z_1,Z_2) = (z_1,z_2)])^{-1} I(F_{1, \mathcal{L}_1}  U_1; Y^n U_0| X_1^n   \mathbf{A}_1 Z_1 Z_2)\\ \nonumber
	& =0,
\end{align}
where the last equality holds by  \eqref{eqM}.
\end{proof}
Next, using Lemmas \ref{lema} and \ref{lem}, we prove Lemma \ref{lem1}.
\begin{lem} [{\cite[Lemma 3]{ahlswede2013oblivious}}] \label{lem}
Let $X,Y,Z$	be random variables defined over the finite sets $\mathcal{X},\mathcal{Y},\mathcal{Z}$, respectively. For any $z_1,z_2 \in \mathcal{Z}$ such that $p \triangleq \mathbb{P}[Z=z_1]>0$ and $q \triangleq \mathbb{P}[Z=z_2]>0$, we have
\begin{multline*}
	|H(X|Y,Z=z_1) - H(X|Y,Z=z_2)| \\ \leq 1+3 \log |\mathcal{X}| \sqrt{\frac{(p+q)\ln 2}{2pq}I(XY;Z)}.
\end{multline*}
\end{lem}

\begin{lem} \label{lem1}
We have 
\begin{align}
	H( F_{1, \mathcal{L}_1 \backslash \{{Z}_1\}}    | X_1^n   \mathbf{A} {Z}_1 {Z}_2 ) =o(n),  \label{eqlem11}\\
 H( F_{2, \mathcal{L}_2 \backslash \{{Z}_2\}}  | X_2^n   \mathbf{A} {Z}_1 {Z}_2 )=o(n). \label{eqlem12}
\end{align}
\end{lem}
\begin{proof}
It is sufficient to prove \eqref{eqlem11}, as the proof of \eqref{eqlem12} can be obtained by exchanging the roles of the servers. For any $z_1, z_2 \in \mathcal{L}_1 \times \mathcal{L}_2$, we write $\mathcal{L}_1 \backslash \{ z_1\} = \{\gamma_i : i \in \llbracket 1 , L_1-1 \rrbracket \}$, and we have
\begin{align}
	& H(F_{1, \mathcal{L}_1 \backslash \{z_1\}}  | X_1^n  \mathbf{A},Z_1 =z_1,Z_2 =z_2) \nonumber\\ \nonumber
	& \leq H(F_{1, \mathcal{L}_1 \backslash \{z_1\}} | X_1^n  \mathbf{A}_1,Z_1 =z_1,Z_2 =z_2) \nonumber\\ \nonumber
	& \stackrel{(a)} \leq H(F_{1, \mathcal{L}_1 \backslash \{z_1\}} | X_1^n  \mathbf{A}_1, Z_1 =\gamma_1, Z_2 = {z}_2) + 1\\
&\phantom{-}+ O\left( nR_1  \sqrt{ I (F_{1, \mathcal{L}_1 \backslash \{z_1\}}  X_1^n   \mathbf{A}_1;Z_1Z_2)} \right)\nonumber\\ \nonumber
	& \stackrel{(b)} \leq H(F_{1, \mathcal{L}_1 \backslash \{z_1\}}| X_1^n  \mathbf{A}_1 , Z_1 = \gamma_1, Z_2 = {z}_2) + o(n)\\ \nonumber
		& \stackrel{(c)} \leq H(F_{1, \mathcal{L}_1 \backslash \{z_1,\gamma_1\}}| X_1^n  \mathbf{A}_1 , Z_1 =\gamma_1, Z_2 = {z}_2) \\
&\phantom{-}+ H(F_{1, \gamma_1}| X_1^n  \mathbf{A}_1 , Z_1 =\gamma_1, Z_2 = {z}_2) + o(n)\nonumber \\ \nonumber
						& \stackrel{(d)} \leq \sum_{i=1}^{L_1-1} H(F_{1, \gamma_i}| X_1^n  \mathbf{A}_1 , Z_1 =\gamma_i, Z_2 = {z}_2) + o(n)\\ \nonumber
	& \stackrel{(e)} = \sum_{i=1}^{L_1-1} H(F_{1, \gamma_i} | Y^n U_0 X_1^n  \mathbf{A}_1, Z_1 =\gamma_i, Z_2 = {z}_2) + o(n)\\ \nonumber 
	& \stackrel{(f)} \leq \sum_{i=1}^{L_1-1} H(F_{1, \gamma_i}  | Y^n U_0 \mathbf{A}_1 \widehat{F}_{1,\gamma_i}  , Z_1 =\gamma_i, Z_2 = {z}_2)+o(n) \\ 
	& \stackrel{(g)} \leq o(n), \label{eqh2}
\end{align}
where 
\begin{enumerate}[(a)]
    \item holds by Lemma \ref{lem};
    \item holds by \eqref{eqpra};
    \item holds by the chain rule and because conditioning reduces entropy;
    \item holds by repeating $L_1-2$ times the steps between $(a)$ and $(c)$;
    \item holds by Lemma \ref{lema};
    \item holds because for any $i \in \llbracket 1 , L_1-1 \rrbracket$, $\widehat{F}_{1,\gamma_i}$ is a function of $(Y^n, U_0, \mathbf{A}_1)$;
    \item holds by Fano's inequality and \eqref{eqre}. 
    \end{enumerate}
    Finally, we have  
\begin{align*}
&H(F_{1, \mathcal{L}_1 \backslash \{Z_1\}}    | X_1^n \mathbf{A}_1 Z_1 Z_2)\\
& =   \textstyle\sum_{z_1,z_2} \mathbb{P}[(Z_1,Z_2) = (z_1,z_2)] \\
& \phantom{-----} \times H(F_{1, \mathcal{L}_1 \backslash \{Z_1\}}   | X_1^n  \mathbf{A}_1 ,Z_1=z_1, Z_2=z_2)\\
& = o(n),
\end{align*}
where the last equality holds by~\eqref{eqh2}.

\end{proof}
Next, using Lemma \ref{lem1} we obtain the following lemma.
\begin{lem} \label{lemb}
We have
	\begin{align*}
	 H(F_{1, \mathcal{L}_1 \backslash \{{Z}_1\}} F_{2, \mathcal{L}_2 \backslash \{{Z}_2\}} | {Z}_1 {Z}_2 )\leq H(X_1^n X_2^n | Y^n ) + o(n).
	 \end{align*}
\end{lem}

\begin{proof}
	We have 	
\begin{align*}
  & H(F_{1, \mathcal{L}_1 \backslash \{{Z}_1\}} F_{2, \mathcal{L}_2 \backslash \{{Z}_2\}}  | {Z}_1 {Z}_2 ) \\
	&= H(F_{1, \mathcal{L}_1 \backslash \{{Z}_1\}} F_{2, \mathcal{L}_2 \backslash \{{Z}_2\}} | Y^n  \mathbf{A} {Z}_1 {Z}_2 ) \\
    &\phantom{-}+ I(F_{1, \mathcal{L}_1 \backslash \{{Z}_1\}} F_{2, \mathcal{L}_2 \backslash \{{Z}_2\}} ; Y^n  \mathbf{A}  | {Z}_1 {Z}_2 )\\
	&\leq H(F_{1, \mathcal{L}_1 \backslash \{{Z}_1\}} F_{2, \mathcal{L}_2 \backslash \{{Z}_2\}} | Y^n  \mathbf{A} {Z}_1 {Z}_2 ) \\
    &\phantom{-}+ I(F_{1, \mathcal{L}_1 \backslash \{{Z}_1\}} F_{2, \mathcal{L}_2 \backslash \{{Z}_2\}} ; Y^n  \mathbf{A}   {Z}_1 {Z}_2 )\\
	&\stackrel{(a)} \leq  H(F_{1, \mathcal{L}_1 \backslash \{{Z}_1\}} F_{2, \mathcal{L}_2 \backslash \{{Z}_2\}} | Y^n  \mathbf{A} {Z}_1 {Z}_2 )  + o(n)\\
	&\leq H(F_{1, \mathcal{L}_1 \backslash \{{Z}_1\}} F_{2, \mathcal{L}_2 \backslash \{{Z}_2\}} X_1^n X_2^n| Y^n \mathbf{A} {Z}_1 {Z}_2 ) + o(n)\\
	& = H( X_1^n X_2^n| Y^n  \mathbf{A} {Z}_1 {Z}_2 ) \\
    &\phantom{-}+ H(F_{1, \mathcal{L}_1 \backslash \{{Z}_1\}} F_{2, \mathcal{L}_2 \backslash \{{Z}_2\}}| X_1^n X_2^n Y^n  \mathbf{A} {Z}_1 {Z}_2 )+o(n)\\
	& \stackrel{(b)} \leq  H(X_1^n X_2^n | Y^n ) + H(F_{1, \mathcal{L}_1 \backslash \{{Z}_1\}}   | X_1^n   \mathbf{A} {Z}_1 {Z}_2 ) \\
    &\phantom{-}+ H( F_{2, \mathcal{L}_2 \backslash \{{Z}_2\}} |  X_2^n   \mathbf{A} {Z}_1 {Z}_2 ) +o(n)\\
	& \stackrel{(c)} \leq  H(X_1^n X_2^n | Y^n ) + o(n),
\end{align*}
where \begin{enumerate}[(a)]
    \item holds by \eqref{eqpr2};
    \item holds by the chain rule and because conditioning reduces entropy;
    \item holds  by Lemma \ref{lem1}.
    \end{enumerate}
\end{proof}
Finally, we have
\begin{align*}
	&(L_1 -1)nR_1 + (L_2-1)nR_2 \\
	&\stackrel{(a)} = H(F_{1, \mathcal{L}_1 \backslash \{{Z}_1\}} F_{2, \mathcal{L}_2 \backslash \{{Z}_2\}} | {Z}_1 {Z}_2 )\\
	& \stackrel{(b)} \leq  H(X_1^n X_2^n | Y^n ) + o(n)\\
	& \stackrel{(c)} \leq \sum_{t=1}^n H( (X_1)_t (X_2)_t | Y_t) + o(n) \\
	& \stackrel{(d)} =  n H( (X_1)_T (X_2)_T | Y_T T) + o(n)\\
	& \leq n H( (X_1)_T (X_2)_T | Y_T) + o(n)\\
	& \stackrel{(e)}\leq n \max_{p_{X_1}p_{X_2}}H(X_1X_2|Y) + o(n),
\end{align*}
where 
\begin{enumerate}[(a)]
    \item holds by independence and uniformity of the files;
    \item holds by Lemma \ref{lemb};
    \item holds by the chain rule and because conditioning reduces entropy;
    \item holds by defining $T$ as the uniform random variable over $\llbracket 1,n\rrbracket$;
    \item holds with the definition $Y \triangleq X_1 + X_2$.
    \end{enumerate}
\end{proof}

Finally, by Proposition \ref{th1}, it is sufficient to use the following lemma to obtain the converse part of Theorem \ref{th2}.
\begin{lem} \label{lemopt}
Consider the random variable $Y$ taking values over $ \mathcal{Y} \triangleq \{ 0,1,2\}$ and defined by
$
Y \triangleq X_1 + X_2,$
where $X_1$ and $X_2$ are independent random variables taken values in $ \mathcal{X} \triangleq \{ 0,1\}$ with probability distribution $p_{X_1}$ and $p_{X_2}$, respectively. We have
\begin{align} 
\max_{p_{X_1}p_{X_2}}H(X_1X_2|Y) = \frac{1}{2}. \label{lemsymeq}
\end{align}
\end{lem}
Although the optimization problem in \eqref{lemsymeq} has symmetry, Lemma \ref{lemopt} is not trivial as symmetric optimization problems may not always have symmetric solutions.

\begin{proof}
Define $p_1 \triangleq p_{X_1}(1)$, $p_2 \triangleq p_{X_2}(1)$. We proceed in four steps.\\
\textbf{Step 1}. We have 
\begin{align*}
H(X_1X_2|Y) 
& = H(X_1|Y) + H(X_2|YX_1)\\
& \stackrel{(a)}= H(X_1|Y)\\
& \stackrel{(b)}=  H(X_1|Y=1)\mathbb{P} [Y=1] ,
\end{align*}
where $(a)$ holds because $X_2 = Y-X_1$, $(b)$ holds because   $Y=0 \implies X_1 =0$. Next, remark that $\mathbb{P}[X_1 =1|Y=1] = (p_1 \bar{p}_2 )/(p_1 \bar{p}_2 + \bar{p}_1 p_2)$ and $\mathbb{P}[Y=1] =p_1 \bar{p}_2 + \bar{p}_1 p_2$. Hence, we can rewrite the optimization problem in the left-hand side of \eqref{lemsymeq} as follows
\begin{align} \label{eqmax}
\max_{p_{X_1}p_{X_2}}H(X_1X_2|Y)
 = \max_{ (p_1, p_2) \in [0,1] \times [0,1]} f(p_1,p_2)  ,
\end{align}
where 
\begin{align*}
&f:]0,1[ \times ]0,1[ \to \mathbb{R}_+,\\
&(p_1,p_2) \mapsto    -\bar{p}_1 p_2 \log \frac{\bar{p}_1 p_2}{p_1 \bar{p}_2 + \bar{p}_1 p_2} -p_1 \bar{p}_2  \log \frac{p_1 \bar{p}_2 }{p_1 \bar{p}_2 + \bar{p}_1 p_2},
\end{align*}
which is continuously extended when $p_1,p_2 \in \{ 0,1\}$.\\ 
\textbf{Step 2}. We now prove that a necessary  condition for the maximum in \eqref{eqmax} to be achieved is $p_1 + p_2 =1$. After computations, we have 
\begin{multline} \label{eqlog}
\frac{\partial f}{\partial p_1}(p_1,p_2) + \frac{\partial f}{\partial p_2}(p_1,p_2) \\=  (p_1 + p_2 -1) \log \frac{p_1 \bar{p}_1 p_2 \bar{p}_2 }{(p_1 \bar{p}_2 + \bar{p}_1 p_2)^2}.
\end{multline}
We remark that the fraction in the $\log$ in \eqref{eqlog} is equal to $1$ if and only if 
\begin{align*}
0 & =  (p_1 \bar{p}_2 + \bar{p}_1 p_2)^2 - p_1 \bar{p}_1 p_2 \bar{p}_2\\
& =  (p_1 \bar{p}_2 )^2 + ( \bar{p}_1 p_2)^2 + p_1 \bar{p}_1 p_2 \bar{p}_2, 
\end{align*}
which is true if and only if $(p_1,p_2) \in \{  (0,0), (1,1)\}$, which would yield $H(X_1X_2|Y=1)=0$. Hence, we conclude that the maximum in \eqref{eqmax} is achieved when $p_1 + p_2 - 1 =0$.\\
\textbf{Step 3}. Define for $g : [0,1] \to \mathbb{R}_+, p_1 \mapsto f(p_1 ,1-p_1)$. Since for any $p_1 \in [0,1]$, $g(p_1) = g(1-p_1)$, by Step~2, we have 
\begin{align} \label{eqmax1}
\max_{p_{X_1}p_{X_2}}H(X_1X_2|Y)
 = \max_{ p_1 \in [0,1/2]} g(p_1).
\end{align}
\textbf{Step 4}. We now show that  $\max_{ p_1 \in [0,1/2]} g(p_1)=g(1/2)$. 
Specifically, it is sufficient to show that $g$ is non-decreasing over $]0,1/2]$ by showing that the derivative $g'$ of $g$ is non-negative over $]0,1/2]$. We have for $x \in ]0,1/2]$, $g'(x) =  2(1-x) \log \frac{(1-x)^2}{(1-x)^2+x^2}- 2x \log \frac{x^2}{(1-x)^2+x^2}$. Next, define $k:[1, + \infty[ \to \mathbb{R}, y \mapsto \log (1+y^2) - y \log (1+y^{-2})$ such that, by the change of variable $y \leftarrow (x^{-1}-1)$, we have $\forall x \in ]0,1/2],g'(x)\geq 0 \Leftrightarrow \forall y \in [1,+\infty[, k(y)\geq 0$. Since $k(1)=0$, to show that $\forall y \in [1,+\infty[, k(y)\geq 0$, it is sufficient to show that $k$ is non-decreasing over $[1,+\infty[$ by showing that its derivative is non-negative over $[1,+\infty[$. We have for  $y\in [1,+\infty[$, $k'(y) = \frac{2}{\ln 2} \frac{1+y}{1+y^2} - \log (1 + y^{-2})$. Since $k'(1) =  \frac{2 - \ln 2}{\ln 2} >0$ and $\lim_{y\to + \infty}k'(y) = 0$, to show that $k'$ is non-negative over $[1,+\infty[$, it is sufficient to show that $k'$ is non-increasing over $[1,+\infty[$, which is true since the derivative $k^{''}$ of $k'$ satisfies for $y\in [1,+\infty[$, $k^{''}(y) = \frac{2(1-y)(1+y)^2}{y(1+y^2)^2 \ln 2} \leq 0$.

All in all, we conclude that  $\max_{ p_1 \in [0,1/2]} = g(1/2) = 1/2$ and, by \eqref{eqmax1}, we obtain the lemma.
\end{proof}

\section{Concluding remarks} \label{secconc}
We studied information-theoretically secure SPIR in the absence of shared randomness, a noisy channel, and data replication. Instead, we leveraged a noiseless binary adder multiple-access channel and two non-colluding servers with independent content to provide information-theoretic security. We fully characterized the capacity region for this setting. While we assumed that all parties strictly follow the protocol, an open problem is to address potential attempts by parties to cheat, for instance, as considered in \cite{pinto2011achieving,dowsley2017oblivious,nascimento2008oblivious}. Other open problems include considering that the client has side information about the files stored in the servers,  for instance, as considered in \cite{chen2020capacity,8882304,8903451,9216062,9335998,10061385,zivarifard2024private,zivarifard2025private,Shekofteh2025A},  investigating achievable rates when only limited communication is allowed between the servers, and characterizing  the capacity region for more than two servers. For the latter research direction, a first challenge is extending the converse, specifically generalizing Lemma~\ref{lemopt} to more than two inputs. A second challenge arises in the achievability, in defining erasure sets that enable simultaneous private retrieval from all servers. Specifically, we conjecture that reducing the multi-file problem to a two-file-per-server scenario and repeatedly applying the two-file algorithm may not be optimal.

\appendices
 \section{One-time pad lemma} \label{applemotp}

\begin{lem} [One-time pad]\label{lemotp}
Let $\epsilon, \delta \geq 0$. Consider the random variables $A$, $B$, $C$, $D$ such that $C$ and $D$ are defined over $\{0,1 \}^m$, $H(D) \geq m - \delta$, and $I(D;ABC)\leq \epsilon$. Then, we have  
\begin{align}
I(A;B,C \oplus D) & \leq  I(A;B) + \epsilon +\delta. \label{eqotp1}
\end{align}
Moreover, if $I(A;C)=0$, then 
\begin{align}
I( A,C\oplus D;C) &\leq \epsilon +\delta. \label{eqotp2}
\end{align}
\end{lem}
\begin{proof}
Equation \eqref{eqotp1} is proved as follows. We have 
\begin{align}
&I(A;B,C \oplus D) \nonumber \\
& = I(A;B) + I(A;C \oplus D|B) \nonumber \\
& \leq I(A;B) + I(ABC;C \oplus D) \label{eqstepd} \\  \nonumber
&  = I(A;B) + H(C \oplus D) - H(C \oplus D|ABC) \\  \nonumber
&  = I(A;B) + H(C \oplus D) - H(D|ABC) \\  \nonumber
& \stackrel{(a)} \leq I(A;B) + H(C \oplus D) - m + \delta+ \epsilon \\
&  \stackrel{(b)} \leq I(A;B)+ \epsilon +\delta,\label{eqstepf}
\end{align}
where $(a)$ holds because $I(D;ABC)\leq \epsilon$ and $H(D) \geq m - \delta$, $(b)$ holds because $H(C \oplus D) \leq m$. Next, we have
\begin{align*}
I(A,C \oplus D;C)
& = I(A;C) + I(C \oplus D;C|A) \nonumber \\
& \stackrel{(a)}\leq I(C \oplus D;AC)  \\  \nonumber
&  \stackrel{(b)} \leq \epsilon +\delta,
\end{align*}
where $(a)$ holds because $I(A;C)=0$, and $(b)$ holds  by choosing $B \leftarrow \emptyset$ in the successive steps between \eqref{eqstepd} and~\eqref{eqstepf}.
\end{proof}

\section{Proof of Lemma \ref{lem1a}} \label{App_lem1}

We have
\begin{align*}
& I(Z_{1} Z_{2} Y^{n}_{1:K} \mathbf{A}_{1:K}; \bar{C}_{1,1:K}   \bar{C}_{2,1:K}   | W_{1:K} ) \\
& =  I(Z_{1} Z_{2}  ; \bar{C}_{1,1:K}   \bar{C}_{2,1:K}   | W_{1:K} )  \\
& \phantom{-}+ I( Y^{n}_{1:K} \mathbf{A}_{1:K}; \bar{C}_{1,1:K}   \bar{C}_{2,1:K}   | Z_{1} Z_{2} W_{1:K} ) \\
& \stackrel{(a)}= \sum_{t=1}^{K} I( Y^{n}_{t} \mathbf{A}_{t}; \bar{C}_{1,1:K}   \bar{C}_{2,1:K}   |Z_1Z_2  Y^{n}_{1:t-1} \mathbf{A}_{1:t-1} W_{1:K} ) \\
& \stackrel{(b)}\leq \sum_{t=1}^{K} H(Y^{n}_{t} \mathbf{A}_{t} |Z_1Z_2   W_{t} ) \\
& \phantom{--} - H( Y^{n}_{t} \mathbf{A}_{t} | \bar{C}_{1,1:K}   \bar{C}_{2,1:K}  Z_1Z_2  Y^{n}_{1:t-1} \mathbf{A}_{1:t-1} W_{1:K} )\\
& \stackrel{(c)}= \sum_{t=1}^{K} H( Y^{n}_{t} \mathbf{A}_{t} |Z_1Z_2   W_{t} )- H( Y^{n}_{t} \mathbf{A}_{t} | \bar{C}_{1,t}   \bar{C}_{2,t} Z_1Z_2 W_{t} ) \\
&  \stackrel{(d)}= \sum_{t=1}^{K} H( Y^{n}_{t} \mathbf{A}_{t} |Z_1Z_2 Z_{1,t} Z_{2,t}  W_{t} )\\
& \phantom{--}- H( Y^{n}_{t} \mathbf{A}_{t} | \bar{C}_{1,t}   \bar{C}_{2,t} Z_1Z_2 Z_{1,t} Z_{2,t} W_{t} ) \\
&  \stackrel{(e)}\leq \sum_{t=1}^{K} H( Y^{n}_{t} \mathbf{A}_{t} |Z_{1,t} Z_{2,t}  W_{t} )\\
& \phantom{--}- H(Y^{n}_{t} \mathbf{A}_{t} | \bar{C}_{1,t}   \bar{C}_{2,t}  Z_{1,t} Z_{2,t} W_{t} ) \\
& = \sum_{t=1}^{K} I( Y^{n}_{t} \mathbf{A}_{t} ;\bar{C}_{1,t}   \bar{C}_{2,t}|Z_{1,t} Z_{2,t}  W_{t} )  \\
& \leq \sum_{t=1}^{K} I(Z_{1,t} Z_{2,t} Y^{n}_{t} \mathbf{A}_{t} ;\bar{C}_{1,t}   \bar{C}_{2,t}|  W_{t} )  \\
& \stackrel{(f)}= 0, 
\end{align*}
where 
\begin{enumerate} [(a)]
    \item holds by the chain rule and because $I(Z_{1} Z_{2}  ; \bar{C}_{1,1:K}   \bar{C}_{2,1:K}   | W_{1:K} ) =0$;
     \item holds by the definition of the mutual information and because conditioning reduces entropy;
      \item holds because $(Y^{n}_{t} ,\mathbf{A}_{t})  - (\bar{C}_{1,t}  , \bar{C}_{2,t}, Z_1,Z_2  , W_{t})- (\bar{C}_{1,1:K}  , \bar{C}_{2,1:K} , Y^{n}_{1:t-1} \mathbf{A}_{1:t-1} W_{1:K})$ forms a Markov chain; 
       \item holds because for any $t \in \llbracket 1 , K \rrbracket$, $(Z_{1,t}, Z_{2,t} )$ is a function of $(Z_1,Z_2)$;  \item holds because conditioning reduces entropy and because $( Y^{n}_{t}, \mathbf{A}_{t}) - ( \bar{C}_{1,t} ,  \bar{C}_{2,t} , Z_{1,t}, Z_{2,t}, W_{t}) - (Z_1,Z_2)$ forms a Markov chain; 
\item holds by the servers' privacy with respect to the client in Line 7 of Algorithm~\ref{algrtoc}.
\end{enumerate}

\section{Proof of Lemma \ref{lem2a}} \label{App_lem2}
We have \eqref{eqtop3}, 
\begin{figure*}
\begin{align*} 
& I(W_{1:K} ; F_{1, \mathcal{L}_1}F_{2, \mathcal{L}_2} | F_{1,Z_1} F_{2,Z_2} ) \\
& \leq  I(W_{1:K} Z_{1,1:K} Z_{2,1:K}; F_{1, \mathcal{L}_1}F_{2, \mathcal{L}_2} | F_{1,Z_1} F_{2,Z_2} ) \\
&\stackrel{(a)} =  I( W_{1:K}; F_{1, \mathcal{L}_1}F_{2, \mathcal{L}_2} | F_{1,Z_1} F_{2,Z_2}Z_{1,1:K} Z_{2,1:K})\\
& =  H( W_{1:K}| F_{1,Z_1} F_{2,Z_2}  Z_{1,1:K}Z_{2,1:K})  -H( W_{1:K} | F_{1, \mathcal{L}_1}F_{2, \mathcal{L}_2}  
   Z_{1,1:K}Z_{2,1:K} ) \\
&\stackrel{(b)} \leq   H( (C_{1,t,i}[Z_{1,t,i}])_{t \in \llbracket 1,L_1-1 \rrbracket \backslash \{ Z_1\} , i\in \llbracket 1,L_2-1 \rrbracket  } | F_{1,Z_1} F_{2,Z_2}  Z_{1,1:K}Z_{2,1:K} ) \\&\phantom{--}+ H( (C_{2,t,i}[Z_{2,t,i}])_{t \in \llbracket 1,L_2-1 \rrbracket \backslash \{ Z_2\} , i\in \llbracket 1,L_1-1 \rrbracket  } | F_{1,Z_1} F_{2,Z_2}  Z_{1,1:K}Z_{2,1:K} )\\
& \phantom{--} -H(  (S_{1,t,i})_{t \in \llbracket 1,L_1-2 \rrbracket, i\in \llbracket 1,L_2-1 \rrbracket  } (S_{2,t,j})_{t \in \llbracket 1,L_2-2 \rrbracket, j\in \llbracket 1,L_1-1 \rrbracket  }  | F_{1, \mathcal{L}_1}F_{2, \mathcal{L}_2}  Z_{1,1:K}Z_{2,1:K}) \\
& \stackrel{(c)}\leq  (L_1-2)(L_2-1) \frac{nR_1}{L_2-1}+ (L_2-2)(L_1-1)\frac{nR_2}{L_1-1}  -H(  (S_{1,t,i})_{t \in \llbracket 1,L_1-2 \rrbracket, i\in \llbracket 1,L_2-1 \rrbracket  } (S_{2,t,j})_{t \in \llbracket 1,L_2-2 \rrbracket, j\in \llbracket 1,L_1-1 \rrbracket  }) \\
& = 0 \numberthis \label{eqtop3}
\end{align*}
\hrulefill
\end{figure*}
where 
\begin{enumerate} [(a)]
    \item holds by the chain rule and because $I(Z_{1,1:K}Z_{2,1:K}; F_{1, \mathcal{L}_1}F_{2, \mathcal{L}_2} | F_{1,Z_1} F_{2,Z_2} ) = 0$; 
    \item holds by definition of $(W_t)_{t \in \llbracket 1 , K \rrbracket}$; 
    \item holds because conditioning reduces entropy, and by independence between  $((S_{1,t,i})_{t \in \llbracket 1,L_1-2 \rrbracket, i\in \llbracket 1,L_2-1 \rrbracket  } ,(S_{2,t,j})_{t \in \llbracket 1,L_2-2 \rrbracket, j\in \llbracket 1,L_1-1 \rrbracket  }) $ and $(F_{1, \mathcal{L}_1},F_{2, \mathcal{L}_2},  Z_{1,1:K},Z_{2,1:K})$.
\end{enumerate}

\section{Proof of Lemma \ref{lem3a}} \label{App_lem3}
We have
\begin{align*}
&\sum_{t=1}^{K}  I( \bar{C}_{1,t}  X^n_{1,t} \mathbf{A}_{t}; Z_{1,t}Z_{2,t}| \bar{C}_{1,1:t-1}  X^n_{1,1:t-1} \mathbf{A}_{1:t-1} ) \\
& \stackrel{(a)}\leq \sum_{t=1}^{K}  H( \bar{C}_{1,t}  X^n_{1,t} \mathbf{A}_{t}) \\
& \phantom{--}- H( \bar{C}_{1,t}  X^n_{1,t} \mathbf{A}_{t} | Z_{1,t}Z_{2,t} \bar{C}_{1,1:t-1}  X^n_{1,1:t-1} \mathbf{A}_{1:t-1} ) \\
& \stackrel{(b)}= \sum_{t=1}^{K}  H( \bar{C}_{1,t}  X^n_{1,t} \mathbf{A}_{t}) - H( \bar{C}_{1,t}  X^n_{1,t} \mathbf{A}_{t} | Z_{1,t}Z_{2,t} ) \\
& = \sum_{t=1}^{K}  I( \bar{C}_{1,t}  X^n_{1,t} \mathbf{A}_{t} ; Z_{1,t}Z_{2,t} ) \\
&  \stackrel{(c)}= 0, \numberthis \label{eqlim1m}
\end{align*}
where 
\begin{enumerate} [(a)]
    \item holds by the definition of mutual information and because conditioning reduces entropy; 
    \item holds by the Markov chain $ (\bar{C}_{1,t} , X^n_{1,t} ,\mathbf{A}_{t} )- (Z_{1,t},Z_{2,t} ) - (\bar{C}_{1,1:t-1},  X^n_{1,1:t-1}, \mathbf{A}_{1:t-1} )$;
    \item holds by the privacy of the client file selection.
\end{enumerate}
Then, we have \eqref{eqtop4}
\begin{figure*}
\begin{align*}
&   I( \bar{C}_{1,1:K}  X^n_{1,1:K} \mathbf{A}_{1:K};\bar{C}_{2,1:K}) \\
& = \sum_{t=1}^{K}  I( \bar{C}_{1,t}  X^n_{1,t} \mathbf{A}_{t};\bar{C}_{2,1:K} | \bar{C}_{1,1:t-1}  X^n_{1,1:t-1} \mathbf{A}_{1:t-1} ) \\
& \leq \sum_{t=1}^{K}  I( \bar{C}_{1,t}  X^n_{1,t} \mathbf{A}_{t};\bar{C}_{2,1:K} Z_{1,t}Z_{2,t}| \bar{C}_{1,1:t-1}  X^n_{1,1:t-1} \mathbf{A}_{1:t-1} ) \\
& = \sum_{t=1}^{K}  I( \bar{C}_{1,t}  X^n_{1,t} \mathbf{A}_{t}; Z_{1,t}Z_{2,t}| \bar{C}_{1,1:t-1}  X^n_{1,1:t-1} \mathbf{A}_{1:t-1} )  + I( \bar{C}_{1,t}  X^n_{1,t} \mathbf{A}_{t};\bar{C}_{2,1:K} | Z_{1,t}Z_{2,t} \bar{C}_{1,1:t-1}  X^n_{1,1:t-1} \mathbf{A}_{1:t-1} ) \\
& \stackrel{(a)} = \sum_{t=1}^{K} I( \bar{C}_{1,t}  X^n_{1,t} \mathbf{A}_{t};\bar{C}_{2,1:K} | Z_{1,t}Z_{2,t} \bar{C}_{1,1:t-1}  X^n_{1,1:t-1} \mathbf{A}_{1:t-1} )\\
& = \sum_{t=1}^{K}  [I( \bar{C}_{1,t}  X^n_{1,t} \mathbf{A}_{t};\bar{C}_{2,t} | Z_{1,t}Z_{2,t}  \bar{C}_{1,1:t-1}  X^n_{1,1:t-1} \mathbf{A}_{1:t-1} ) + I( \bar{C}_{1,t}  X^n_{1,t} \mathbf{A}_{t};\bar{C}_{2,1:K}|Z_{1,t}Z_{2,t}  \bar{C}_{2,t}  \bar{C}_{1,1:t-1}  X^n_{1,1:t-1} \mathbf{A}_{1:t-1} ) ]\\
& \stackrel{(b)}\leq \sum_{t=1}^{K}  [I( \bar{C}_{1,t}  X^n_{1,t} \mathbf{A}_{t};\bar{C}_{2,t}   Z_{1,t}Z_{2,t} ) + I( \bar{C}_{1,t}  X^n_{1,t} \mathbf{A}_{t};\bar{C}_{2,1:K}|Z_{1,t}Z_{2,t} \bar{C}_{2,t}  \bar{C}_{1,1:t-1}  X^n_{1,1:t-1} \mathbf{A}_{1:t-1} ) ]\\
& = \sum_{t=1}^{K}  I( \bar{C}_{1,t}  X^n_{1,t} \mathbf{A}_{t};\bar{C}_{2,t}   Z_{1,t}Z_{2,t} )\\
&= \sum_{t=1}^{K}  I( \bar{C}_{1,t}  X^n_{1,t} \mathbf{A}_{t};\bar{C}_{2,t}|    Z_{1,t}Z_{2,t} ) + I( \bar{C}_{1,t}  X^n_{1,t} \mathbf{A}_{t};   Z_{1,t}Z_{2,t} )\\
& \stackrel{(c)}=  0 \numberthis \label{eqtop4} 
\end{align*}
\hrulefill
\end{figure*}
where
\begin{enumerate}[(a)]
    \item holds by \eqref{eqlim1m}; 
    \item holds because 
\begin{align*}
& I( \bar{C}_{1,t}  X^n_{1,t} \mathbf{A}_{t};\bar{C}_{2,t} | Z_{1,t}Z_{2,t}  \bar{C}_{1,1:t-1}  X^n_{1,1:t-1} \mathbf{A}_{1:t-1} )\\
& \leq H( \bar{C}_{1,t}  X^n_{1,t} \mathbf{A}_{t} ) \\
&\phantom{-}- H( \bar{C}_{1,t}  X^n_{1,t} \mathbf{A}_{t}|Z_{1,t}Z_{2,t}  \bar{C}_{2,t}  \bar{C}_{1,1:t-1}  X^n_{1,1:t-1} \mathbf{A}_{1:t-1} )\\
& = H( \bar{C}_{1,t}  X^n_{1,t} \mathbf{A}_{t} ) - H( \bar{C}_{1,t}  X^n_{1,t} \mathbf{A}_{t}|Z_{1,t}Z_{2,t} \bar{C}_{2,t} )\\
&= I( \bar{C}_{1,t}  X^n_{1,t} \mathbf{A}_{t};\bar{C}_{2,t} Z_{1,t}Z_{2,t} ),
\end{align*}
 where the first inequality holds by the definition of the mutual information and because  conditioning reduces entropy, and the  first equality holds because $(\bar{C}_{1,t} , X^n_{1,t} ,\mathbf{A}_{t})  - (Z_{1,t},Z_{2,t},\bar{C}_{2,t}) -(  \bar{C}_{1,1:t-1},  X^n_{1,1:t-1}, \mathbf{A}_{1:t-1})$ forms a Markov chain; 
    \item holds by the servers' mutual privacy, and the client privacy with respect to the servers. 
 \end{enumerate}

\bibliographystyle{IEEEtran}
\bibliography{bib}

\end{document}